\documentclass[aps,pra,reprint,groupedaddress,longbibliography,amsmath,amssymb,floatfix]{revtex4-2}
\usepackage{amsmath,amssymb,amsthm,mathtools,bm}
\usepackage[T1]{fontenc}
\usepackage{microtype}
\usepackage[colorlinks=true,linkcolor=blue!50!black,urlcolor=blue!50!black,citecolor=blue!50!black]{hyperref}
\usepackage{xcolor}
\usepackage{tikz}
\usetikzlibrary{arrows.meta,calc,decorations.pathmorphing,patterns}

\definecolor{dividerGray}{HTML}{E5ECF2}
\definecolor{probeBlue}{HTML}{2878B5}
\definecolor{sensorTeal}{HTML}{168B89}
\definecolor{signalOrange}{HTML}{E58A28}
\definecolor{preparePurple}{HTML}{8A72B6}
\definecolor{mirrorGray}{HTML}{7D8B98}
\definecolor{ink}{HTML}{24364B}
\definecolor{guideGray}{HTML}{E5ECF2}
\definecolor{zeroGray}{HTML}{EDF3F7}

\newcommand{\renewaloffset}{2}

\newcommand{\backgroundamp}[1]{\zeta_{#1}}
\newcommand{\signalamp}[1]{\alpha_{#1}}
\newcommand{\finalamp}{\zeta_F}

\newtheorem{theorem}{Theorem}
\newtheorem{lemma}{Lemma}
\newtheorem{proposition}{Proposition}

\theoremstyle{definition}
\newcommand{\CN}{\mathcal{CN}}
\newcommand{\E}{\mathbb E}
\newcommand{\Prob}{\Pr}
\newcommand{\Tr}{\operatorname{Tr}}
\newcommand{\TV}{\operatorname{TV}}
\newcommand{\Id}{\mathbb I}
\newcommand{\tc}{\tau_B}
\newcommand{\tobs}{T_{\mathrm{tot}}}
\newcommand{\eps}{\varepsilon}
\newcommand{\ep}{\varepsilon_0}
\newcommand{\lc}{\ell_c}
\newcommand{\ket}[1]{|#1\rangle}
\newcommand{\bra}[1]{\langle#1|}
\newcommand{\norm}[1]{\lVert#1\rVert}
\newcommand{\abs}[1]{\lvert#1\rvert}

\DeclareMathOperator{\Var}{Var}
\DeclareMathOperator{\Cov}{Cov}
\renewcommand{\Re}{\operatorname{Re}}
\renewcommand{\Im}{\operatorname{Im}}

\begin{document}
\title{Exponential Quantum Advantage in Stationary Signal Sensing}
\author{Sojeong Park}
\affiliation{Department of Physics, Korea Advanced Institute of Science and Technology, Daejeon 34141, Korea}
\author{Changhun Oh}
\email{changhun0218@kaist.ac.kr}
\affiliation{Department of Physics, Korea Advanced Institute of Science and Technology, Daejeon 34141, Korea}
\begin{abstract}Sensing time-varying signals is essential for precision measurement and for characterizing fluctuations in physical systems. In this work, we show that quantum entanglement provides an exponential advantage in the acquisition time required to learn the temporal statistics of a stationary classical signal. The task is to estimate the characteristic function of joint displacements on selected bosonic modes within a fixed-duration window, using continuous measurements of one signal realization. We present an entanglement-assisted scheme based on two-mode squeezed probes and local Bell measurements with protected idlers, and bound its acquisition time accounting for readout noise and correlations between signal windows. For a strictly stationary pulse family, it achieves polynomial acquisition time at inverse-polynomial accuracy with logarithmic squeezing and polynomial photon number. On the other hand, we prove that any local idler-free protocols require acquisition time exponential in the number of selected temporal modes per window. This lower bound holds for both sequentially prepared and arbitrarily temporally entangled probes, allowing non-Gaussian inputs and measurements and causal classical feedback, as long as each output is measured destructively within its time bin. These results establish an exponential quantum learning advantage for a single continuously monitored stationary signal, with the sensing cost measured by total acquisition time.
\end{abstract}
\maketitle
\addtocontents{toc}{\protect\iffalse}
\section{Introduction}

Sensing time-varying signals has many applications in precision measurement, including gravitational-wave detection by LIGO~\cite{tse2019}, weak-force sensing with mechanical oscillators~\cite{schreppler2014,iwasawa2013}, and measurements of fluctuating magnetic and optical fields~\cite{degen2017,yonezawa2012}. In these applications, the waveforms and temporal correlations contain information about the dynamics of astrophysical sources, weak interactions, and environmental noise. Therefore, there has been extensive research on quantum sensing of time-varying signals, including adaptive phase tracking and smoothing~\cite{berry2002,tsang2009,wheatley2010}, fundamental limits and protocols for waveform detection and estimation~\cite{tsang2011,tsang2012,miao2017,gardner2024,gardner2025}, and spectrum analysis~\cite{ng2016,tsang2023,sui2026}. Complementary approaches use controlled qubit dynamics for noise spectroscopy~\cite{alvarez2011,bylander2011,szankowski2017}, including higher-order spectra that characterize non-Gaussian fluctuations~\cite{norris2016,sung2019}. Repeated sampling with a quantum sensor can also resolve spectral features beyond its coherence-time limit by using a long measurement record and a stable reference clock~\cite{schmitt2017,boss2017}.

Meanwhile, quantum learning has revealed a broader role for quantum resources in inferring properties of unknown states, processes, and probability distributions from experimental data. Recent theoretical and experimental studies have demonstrated exponential reductions in the required number of samples~\cite{aharonov2022,huang2021,huang2022,chen2022memory,chen2022pauli,oh,liu2025}. In particular, entanglement with a protected reference enables an exponential advantage in learning the characteristic function of a multimode random displacement channel~\cite{oh,liu2025}. Squeezing can also provide a single-mode learning advantage over classical-state probes as the required Fourier resolution increases~\cite{baiju2026}. These results have motivated further studies of sensing correlated noise~\cite{wang2024,prabhu2026,brady2026}, quantum processing of classical fields~\cite{cotler2026,khan2026}, and learning temporal signal correlations~\cite{kannan2026}.

In many sensing applications, the signal is generated by an ongoing physical process beyond the experimenter's control. Access is therefore through a continuous observation record, making it necessary to learn the signal's temporal statistics without resetting the source or preparing independent realizations. Stationary signals, encountered in noise spectroscopy and measurements of fluctuations in a steady operating regime, provide a natural setting for this task. Their joint statistics are invariant under a common time shift, so every window of a fixed duration has the same probability distribution. This provides a common learning target for one continuous record, while correlations between windows must be accounted for in the acquisition time. For pulsed signals, the window distribution includes arbitrary arrival times and partially captured pulses, and the total acquisition time includes waiting intervals between events.

Existing displacement-learning guarantees describe repeated channel uses or independent process realizations~\cite{oh,liu2025,cotler2026}. Learning temporal correlations can also use a long-lived quantum memory to connect specified sensing times within independently sampled signal realizations~\cite{kannan2026}. An advantage in the number of independent samples does not by itself establish an advantage in observation time, because temporal correlations affect how quickly a continuous record supplies new statistical information. Even if well-separated windows are treated as effectively independent samples, the intervening time must be included in the acquisition cost. The central question is whether quantum resources can provide an exponential advantage in the total acquisition time needed to learn the temporal statistics of a single continuously monitored stationary signal, without resets or event alignment. To identify which quantum resource enables this advantage, we also compare entanglement with a protected reference to entanglement among the probes exposed to the signal.

In this work, we establish an exponential separation in total acquisition time for learning the joint temporal statistics of a stationary signal from one continuously measured realization. Our entanglement-assisted scheme uses a two-mode squeezed vacuum (TMSV) and a local continuous-variable Bell measurement on a displaced probe and a protected idler in each time bin. Only classical measurement outcomes are retained across bins. From this record, we construct an unbiased estimator of the window characteristic function and derive an acquisition-time bound that accounts for both readout noise and correlations between signal windows.

We compare this scheme with local idler-free protocols on the same temporal modes, allowing arbitrary probe energy, non-Gaussian inputs and measurements, and causal classical feedback. We consider both sequentially prepared probes and probes entangled across the entire acquisition, with destructive readout in each time bin. For a common strictly stationary pulse family, both classes require acquisition time exponential in the number of selected temporal modes per window, with a stronger bound for sequential preparation. The entanglement-assisted scheme learns the same statistics in polynomial time with logarithmic squeezing and polynomial photon number. Thus, the separation persists even when the probes exposed to the signal may be entangled across time.

\section{Setup}\label{sec:setup}
\subsection{Physical signal and linear bosonic response}

We consider a linear sensor that encodes a classical signal as a displacement of a bosonic field. This description includes optical amplitude and phase modulation, as well as force sensing with an optomechanical resonator in the linear-response regime~\cite{clerk2010,aspelmeyer2014}. We call the prepared quantum input the probe and the output measurement the readout. Together with the sensor dynamics, feedback, and estimation, they specify the sensing protocol. Figure~\ref{fig:optomechanical-response} illustrates this encoding for an optomechanical force sensor: the mechanical and cavity response converts $S_\theta(t)$ into the outgoing field displacement $\xi_\theta(t)$.

Let $S_\theta(t)$ be the classical stochastic signal on the entire time axis, where $\theta$ labels its unknown probability law. A complex signal represents two real drive components coupled to the output quadratures. For a time-independent linear response, the output field is
\begin{align}
 \hat b_\theta(t)=\hat b(t)+\xi_\theta(t),\qquad
 \xi_\theta(t)=\int_{\mathbb R}ds\,h(t-s)S_\theta(s),
 \label{eq:response}
\end{align}
where $h$ is the known causal impulse response, including the susceptibility, coupling, and propagation. The undisplaced output $\hat b(t)$ contains the signal-independent quantum dynamics and satisfies $[\hat b(t),\hat b^\dagger(t')]=\delta(t-t')$. This displacement description follows from linear bosonic input--output theory~\cite{gardiner1985,gough}. The convolution includes the signal history before each analysis window, so successive windows probe the same continuing sensor dynamics.

Throughout this work, we focus on strictly stationary signals and use the stationary solution of the sensor response. This means that the joint distribution of any finite collection of signal values is unchanged by a common real time shift. Specifically, for every positive integer $q$ and all $t_1,\ldots,t_q,a\in\mathbb R$,
\begin{align}
 (S_\theta(t_k+a))_{k=1}^{q}\overset{\mathrm d}{=}(S_\theta(t_k))_{k=1}^{q}.
 \label{eq:strict-stationarity}
\end{align}

Here $\overset{\mathrm d}{=}$ denotes equality in distribution. Since convolution commutes with time translations, $\xi_\theta(t)$ is also strictly stationary. We assume that the convolution is well defined and its output is measurable in time and the signal realization, with finite integrated intensity over every finite interval. These conditions ensure that the temporal-mode integrals below are well defined.

\begin{figure}
\begin{tikzpicture}[
  x=1cm,y=1cm,
  font=\rmfamily\fontsize{8}{9.2}\selectfont,
  text=black,
  line cap=round,line join=round,
  >={Stealth[length=1.5mm,width=1.0mm]},
  light/.style={draw=probeBlue,line width=0.85pt},
  force/.style={draw=signalOrange,line width=0.95pt,->},
  component/.style={draw=probeBlue,fill=white,line width=0.75pt,
    minimum width=1.50cm,minimum height=0.62cm,align=center,
    text=probeBlue,inner sep=1.5pt},
  label/.style={align=center,inner sep=1pt}
]
\pgfmathsetlengthmacro{\optomechanicalfigureunit}{\columnwidth/8.6}
\pgfmathsetlengthmacro{\optomechanicalprobewidth}{1.50*\optomechanicalfigureunit}
\pgfsetxvec{\pgfpoint{\optomechanicalfigureunit}{0pt}}
\pgfsetyvec{\pgfpoint{0pt}{\optomechanicalfigureunit}}

\coordinate (BS) at (2.73,0.55);
\coordinate (Minput) at (4.00,0.55);
\coordinate (Mmove) at (6.70,0.55);

\node[component,minimum width=\optomechanicalprobewidth]
  (probe) at (0.85,0.70) {Probe};
\draw[light,->] (probe.east) -- (3.90,0.70);

\draw[light,->] (3.90,0.40) -- (2.58,0.40);
\node[component,minimum width=1.90cm] (readout) at (2.58,-1.60)
  {Measurement};
\draw[light,->] (2.58,0.40) -- (readout.north);
\node[text=probeBlue,anchor=west,inner sep=1pt] at (2.70,-0.45)
  {$\xi_\theta(t)$};

\draw[draw=mirrorGray,line width=1.6pt]
  ($(BS)+(-0.29,-0.29)$) -- ($(BS)+(0.29,0.29)$);
\node[label,text=black] at (2.73,1.22) {Beam splitter};

\fill[mirrorGray]
  ($(Minput)+(-0.05,-0.80)$) rectangle ($(Minput)+(0.05,0.80)$);
\fill[sensorTeal]
  ($(Mmove)+(-0.07,-0.80)$) rectangle ($(Mmove)+(0.07,0.80)$);
\node[label,text=black,anchor=base] at (4.00,1.84) {Input mirror};
\node[label,text=sensorTeal,anchor=base] at (6.70,1.84) {Moving mirror};

\draw[light,samples=161,smooth,domain=0:1,variable=\u]
  plot ({4.18+2.34*\u},{0.55+0.12*sin(1440*\u)});

\draw[draw=mirrorGray,line width=0.80pt] (6.77,0.55) -- (6.97,0.55);
\draw[draw=mirrorGray,line width=0.80pt,
  decorate,decoration={coil,aspect=0.45,segment length=2.6mm,
                       amplitude=1.3mm,pre length=1.2mm,post length=1.2mm}]
  (6.97,0.55) -- (8.22,0.55);
\fill[pattern=north east lines,pattern color=mirrorGray!70]
  (8.25,0.10) rectangle (8.43,1.00);
\draw[draw=mirrorGray,line width=1.6pt] (8.25,0.10) -- (8.25,1.00);

\draw[force] (8.17,1.13) -- (6.91,1.13);
\node[label,text=signalOrange] at (7.62,1.45) {$S_\theta(t)$};

\node[label,text=sensorTeal] at (6.70,-0.60) {$h(t-s)$};

\pgfresetboundingbox
\path[use as bounding box]
  ($(0,0 |- readout.south)+(0,-0.03)$) rectangle (8.60,2.25);

\end{tikzpicture}
    \caption{Optomechanical realization of the linear-response model. A signal force acts on the moving mirror. The cavity and mechanical dynamics provide a calibrated linear response $h$. The outgoing optical field exits at the input mirror and is routed by a beam splitter into the readout. Its signal-dependent displacement is the causal convolution $\xi_\theta(t)=\int_{\mathbb{R}} h(t-s)S_\theta(s)ds$.}
    \label{fig:optomechanical-response}
\end{figure}

\subsection{Temporal-mode measurement}

In practice, finite detector bandwidth and sampling rates limit the temporal resolution of the readout, even though the sensor output is continuous in time. We describe a readout that collects one temporal mode per interval of duration $\tc>0$, which sets this resolution and the readout rate $1/\tc$. For the intervals $I_j\coloneqq [j\tc,(j+1)\tc)$, we use rectangular profiles, whose normalization gives canonical bosonic modes~\cite{brecht2015,raymer2020}:
\begin{align}
 \psi_j(t)\coloneqq \tc^{-1/2}\mathbf1_{I_j}(t),\qquad
 \hat a_j\coloneqq \int dt\,\psi_j^*(t)\hat b(t).
 \label{eq:interface}
\end{align}
Thus $[\hat a_j,\hat a_k^\dagger]=\delta_{jk}$. The signal displaces this mode by $X_j\in\mathbb C$,
\begin{align}
 X_j\coloneqq \tc^{-1/2}\int_{I_j}\xi_\theta(t)\,dt,
 \qquad \hat D(X_j)\coloneqq e^{X_j\hat a_j^\dagger-X_j^*\hat a_j}.
 \label{eq:cellshift}
\end{align}
Rectangular temporal filters extract modes of continuously emitted fields~\cite{yang2025}; for homodyne detection, integration with this weighting gives the corresponding modal quadrature~\cite{liu2025}. All protocols in the comparison use the same weighting and readout rate throughout acquisition.

We choose a fixed window duration $T$ before acquisition to study temporal statistics such as correlations between different times or joint fluctuations across a signal burst, without aligning the windows to pulse arrivals. For $T=n\tc$, a window starting at any real time $t$ contains the displacement vector $\bm X_t\in\mathbb C^n$, defined by
\begin{align}
 (\bm X_t)_k\coloneqq \tc^{-1/2}\int_{t+k\tc}^{t+(k+1)\tc}\xi_\theta(s)\,ds,
 \quad k=0,\ldots,n-1.
 \label{eq:window}
\end{align}
For windows starting at $t=j\tc$, $(\bm X_{j\tc})_k=X_{j+k}$. Strict stationarity makes $P_{T,\theta}\coloneqq \operatorname{Law}_\theta(\bm X_t)$ independent of $t$. By Eqs.~\eqref{eq:response} and \eqref{eq:window}, each component of $\bm X_t$ is a known linear functional of the physical signal $S_\theta$, determined by the sensor response and temporal-mode profile. The joint statistics of these coordinates describe the temporal correlations and fluctuations of $S_\theta$ accessible through this readout at resolution $\tc$. We therefore take the window distribution $P_{T,\theta}$ as our learning target, including empty windows and partially captured pulses.

\subsection{Learning from continuous measurements}
The protocol measures $J$ consecutive time bins during $[0,\tobs)$, where $\tobs\coloneqq J\tc$ is the total observation time. Our goal is to learn the window distribution $P_{T,\theta}$ by estimating its characteristic function. Since the full characteristic function uniquely determines the distribution, it also describes non-Gaussian statistics. We restrict the estimation task to a finite query range, with $\kappa>0$ setting the squared Fourier scale per selected mode:
\begin{align}
 \chi_{T,\theta}(\beta)\coloneqq \E_{\bm X\sim P_{T,\theta}}
 [e^{\bm X^\dagger\beta-\beta^\dagger\bm X}],
 \quad \abs{\beta}^2\le\kappa n.
 \label{eq:target}
\end{align}
Here $\beta\in\mathbb C^n$ is the query vector. The exponent is $2i\Im(\bm X^\dagger\beta)$; thus, a larger $\abs{\beta}$ resolves finer structure in the displacement distribution.

We benchmark learning of the window distribution through the accuracy of characteristic-function estimates obtained from a common measurement record~\cite{oh}. The sensing protocol is fixed across the allowed range of $\beta$, while the classical estimator $\hat\chi_T(\beta)$ uses $\beta$ to evaluate the desired Fourier component. For every fixed signal law $\theta$ and allowed $\beta$, we require
\begin{align}
 \Prob_\theta[\abs{\hat\chi_T(\beta)-\chi_{T,\theta}(\beta)}\le\eps]
 \ge1-\delta.
 \label{eq:learning}
\end{align}
Here $\theta$ is held fixed, and the probability is over the signal realization, measurement outcomes, and any randomness in the protocol.

To achieve this learning guarantee from one record, we also need to control temporal dependence, since stationarity alone is insufficient. For example, a random offset that remains fixed for all time is strictly stationary, but even an arbitrarily long measurement gives only one draw of the offset. 

We introduce a mixing coefficient to quantify this dependence. It compares the joint distribution of two windows with the product of their individual distributions. Let $\beta_{\rm mix}(g)$ bound the total variation distance (TVD) between these distributions for windows separated by $g$ window steps, uniformly over window origins and the signal family~\cite{bradley2005}. We assume that the sum over all separations is finite and define the dependence time by
\begin{align}
 \tau_{\rm dep}\coloneqq 4T\sum_{g=1}^{\infty}\beta_{\rm mix}(g)<\infty.
 \label{eq:mixsum}
\end{align}The finite sum ensures that correlations between windows do not prevent the estimation error from decreasing as more windows are collected.

\begin{figure}
\begin{tikzpicture}[
  x=1cm,y=1cm,
  font=\rmfamily\fontsize{8}{9.2}\selectfont,text=black,
  inner sep=1pt,
  line cap=round,line join=round,
  >={Stealth[length=1.5mm,width=1.0mm]},
  circuit/.style={draw=black,line width=0.75pt},
  quantum/.style={line width=0.70pt,->},
  panel/.style={font=\rmfamily\fontsize{8}{9.2}\selectfont\bfseries,
    anchor=west},
  note/.style={text=black,align=center},
  displacement/.style={draw=black,fill=white,line width=0.75pt,
    rounded corners=2.5pt,minimum width=1.08cm,minimum height=0.65cm,
    inner sep=1.5pt,outer sep=0pt,
    font=\rmfamily\fontsize{8.5}{9.5}\selectfont},
  continuation/.style={font=\rmfamily\fontsize{12}{12}\selectfont}
]
\pgfmathsetlengthmacro{\sensingfigureunit}{min(8.6cm,\linewidth)/8.6}
\pgfsetxvec{\pgfpoint{\sensingfigureunit}{0pt}}
\pgfsetyvec{\pgfpoint{0pt}{\sensingfigureunit}}
\path[use as bounding box] (0,0) rectangle (8.60,8.60);

\def\firstcol{1.50}
\def\secondcol{3.85}
\def\thirdcol{6.20}
\def\rightedge{8.45}

\node[panel] at (0.12,8.20) {(a) Local idler-free sensing};
\begin{scope}[shift={(0.30,0)}]
\draw[quantum,draw=black] (5.40,8.20) -- (7.40,8.20);
\node[note] at (6.40,8.40) {Time};

\draw[circuit] (\firstcol,7.50) -- (\thirdcol,7.50);
\foreach \j/\x in {1/\firstcol,2/\secondcol,3/\thirdcol}{
  \node at (\x,7.78) {Bin \j};
  \draw[circuit] (\x,7.50) -- (\x,7.25);
  \draw[circuit,fill=white]
    (\x,7.25) -- (\x-0.27,6.90) -- (\x+0.27,6.90) -- cycle;

  \node[displacement] (ifchannel\j) at (\x,6.15) {$D(X_\j)$};
  \draw[circuit] (\x,6.90) -- (ifchannel\j.north);
  \draw[circuit] (ifchannel\j.south) -- (\x,5.40);

  \foreach \offset in {-0.032,0.032}{
    \draw[circuit,line cap=butt]
      (\x+\offset,5.13) -- (\x+\offset,4.62);
  }
  \draw[circuit,fill=white]
    (\x-0.27,5.40) -- (\x+0.27,5.40) -- (\x,5.05) -- cycle;
  \node[anchor=west] at (\x+0.18,4.70) {$o_{\j}$};
}
\node[continuation] at (7.4,6.15) {$\cdots$};

\node at (\secondcol,4.38) {Classical feedback allowed};
\end{scope}

\node[panel] at (0.12,3.65)
  {(b) Entanglement-assisted Bell sensing};
\begin{scope}[shift={(0.30,0)}]

\foreach \j/\x in {1/\firstcol,2/\secondcol,3/\thirdcol}{
  \coordinate (eapair\j) at (\x,3.00);
  \coordinate (earight\j) at (\x+0.72,2.28);
  \node[anchor=base] at (\x,3.18) {TMSV};

  \node[displacement] (eachannel\j) at (\x-0.72,1.85) {$D(X_\j)$};
  \node[draw=black,fill=white,line width=0.75pt,rounded corners=2.5pt,
    minimum width=0.88cm,minimum height=0.46cm,
    inner sep=0pt,outer sep=0pt]
    (bellmeasurement\j) at (\x,0.95) {};
  \node[anchor=east] at (\x-0.60,0.95) {BM};

  \draw[circuit] (eapair\j) -- (\x-0.72,2.28)
    -- (eachannel\j.north);
  \draw[circuit] (eapair\j) -- (earight\j);

  \draw[circuit]
    let \p1=(eachannel\j.south), \p2=(bellmeasurement\j.west),
        \n1={\x2-\x1}
    in (eachannel\j.south) -- (\x1,{\y2+\n1})
       -- (bellmeasurement\j.west);
  \draw[circuit]
    let \p1=(earight\j), \p2=(bellmeasurement\j.east),
        \n1={\x1-\x2}
    in (earight\j) -- (\x1,{\y2+\n1})
       -- (bellmeasurement\j.east);

  \foreach \offset in {-0.032,0.032}{
    \draw[circuit,line cap=butt]
      ($(bellmeasurement\j.south)+(\offset,0)$)
      -- (\x+\offset,0.40);
  }
  \node[anchor=north,inner sep=0pt] at (\x,0.33) {$Y_{\j}$};
}
\node[continuation] at (7.4,1.85) {$\cdots$};
\end{scope}
\end{tikzpicture}
\caption{Local idler-free and entanglement-assisted sensing. Time increases from left to right. Each vertical column represents the operations within one fixed detector time bin. (a) The local scheme permits arbitrary joint input correlations across bins, destructive local measurements, and causal classical feedback, including control of later probe preparation. Each measured output leaves only a classical outcome; quantum correlations among the remaining probes remain available. (b) The EA scheme uses independent TMSV pairs and a protected idler bypassing the signal. BM denotes the Bell measurement within each bin.
}
    \label{fig:local-and-bell-readout}
\end{figure}

\subsection{Sensing protocols}\label{sec:protocols}
Probe preparation addresses the selected temporal modes before signal encoding, and readout measures them at the output. We assume an ideal sensing interface, without loss or additional signal-independent quantum noise. On mode $j$, the signal acts through the ideal displacement channel $\rho\mapsto\hat D(X_j)\rho\hat D^\dagger(X_j)$. We allow joint input states across the full acquisition, so earlier readouts can condition future probes through temporal input correlations.

Each displaced output mode is measured destructively within its interval by an arbitrary local POVM, leaving only its classical outcome in the record; quantum correlations among the remaining probes remain available. In a local idler-free protocol, every accessible probe passes through the signal channel, and no protected ancillary mode is available for a joint measurement. Feedback may use all earlier outcomes and classical randomness to choose operations on future probes and subsequent local measurements. These choices are unrestricted within the fixed temporal-mode interface. Figure~\ref{fig:local-and-bell-readout}(a) depicts the temporally correlated inputs and destructive local readouts.

In the sequentially prepared subclass, each new probe is prepared in an arbitrary state using only the earlier classical outcomes and classical randomness, with no quantum correlations to systems reserved for later bins. After the signal displacement, it is measured destructively by an arbitrary local POVM, which may depend on the same classical history.  Discrete and continuous readouts are allowed (see Appendix~\ref{app:outcomes} for the treatment of continuous outcomes).
Thus, the two idler-free classes differ in whether their inputs can be entangled across time, while both use destructive local measurements. We next consider the performance when a protected idler is available in each time bin.
\section{Entanglement-assisted sensing}\label{sec:upper}
We now present an entanglement-assisted (EA) scheme for the learning task. We first derive the probability distribution of its Bell measurement outcomes and construct an unbiased estimator. We then bound the acquisition time, accounting separately for measurement noise and temporal correlations in the signal.

In every time bin, we prepare an independent two-mode squeezed vacuum (TMSV) with squeezing parameter $r\ge0$~\cite{weedbrook2012}. One mode occupies the selected probe mode and undergoes the signal displacement $X_j$, while the other is kept as an idler isolated from the signal. We then perform a continuous-variable Bell measurement on the probe and idler. This can be implemented by mixing the two modes on a balanced beam splitter and performing homodyne detection of the position quadrature on one output and the momentum quadrature on the other~\cite{braunstein1998,oh,liu2025}. The measurement gives two real outcomes in each bin. For $r>0$, the TMSV correlations reduce the readout noise relative to vacuum inputs, as in quantum-dense metrology~\cite{steinlechner2013}.

To express the measurement outcomes, we use $\hat x\coloneqq \hat a+\hat a^\dagger$ and $\hat p\coloneqq (\hat a-\hat a^\dagger)/i$, with $[\hat x,\hat p]=2i$. Choosing the TMSV phase so that the probe and idler have correlated positions and anticorrelated momenta, the Bell readout measures the commuting combinations $\hat x_a-\hat x_i$ and $\hat p_a+\hat p_i$. The signal shifts their means by $2\Re X_j$ and $2\Im X_j$, respectively. Combining the two measured outcomes as $[x_a-x_i+i(p_a+p_i)]/2$ gives the complex readout $Y_j\in\mathbb C$:
\begin{align}
 Y_j=X_j+\Xi_j,\qquad
 \Xi_j\overset{\rm iid}{\sim}\CN(0,e^{-2r}),
 \label{eq:bell}
\end{align}
where $\CN(0,\Gamma)$ denotes a circular complex Gaussian distribution with covariance $\Gamma$. The Bell noise is independent of the entire signal, with real and imaginary variances $e^{-2r}/2$. This readout runs continuously on every measurement time bin. Figure~\ref{fig:local-and-bell-readout}(b) shows the independent probe--idler pairs and their same-bin Bell readout.

We now construct an estimator from these outcomes. Group the data into $N\coloneqq \lfloor J/n\rfloor$ complete, consecutive windows of duration $T$, and denote their measured vectors by $\bm Y^{(i)}$. Averaging the Gaussian measurement noise multiplies each Fourier feature by $e^{-e^{-2r}\abs{\beta}^2}$. We compensate for this factor and define the estimator
\begin{align} \hat\chi_{T,N}(\beta)\coloneqq \frac{e^{e^{-2r}\abs{\beta}^2}}N\sum_{i=0}^{N-1} e^{(\bm Y^{(i)})^\dagger\beta-\beta^\dagger\bm Y^{(i)}}. \label{eq:estimator}\end{align}
Using Eq.~\eqref{eq:bell}, the noise average cancels the prefactor and leaves the signal Fourier feature in each window~\cite{oh}. Since all windows have the same stationary distribution, we obtain $\E_\theta[\hat\chi_{T,N}(\beta)]=\chi_{T,\theta}(\beta)$. Thus the estimator is unbiased even when the signal windows are correlated.

To determine the acquisition time, we separate the estimator variance into individual-window contributions and covariances between different windows. The Bell noises are independent across windows, so the latter arise solely from signal correlations and are controlled by the mixing condition in Eq.~\eqref{eq:mixsum}. Specifically, each summand in Eq.~\eqref{eq:estimator} has squared modulus $e^{2e^{-2r}\abs{\beta}^2}$, and its covariance with a summand $g$ windows away is bounded in absolute value by $2\beta_{\rm mix}(g)$. Summing these contributions gives
\begin{align}
 \E_\theta[\abs{\hat\chi_{T,N}(\beta)-\chi_{T,\theta}(\beta)}^2]
 \le\frac{e^{2e^{-2r}\kappa n}+\tau_{\rm dep}/T}{N}.
 \label{eq:upper-proof-outline}
\end{align}
The first term is the readout-noise penalty, which squeezing reduces; the second accounts for temporal dependence. Both contributions decrease as $1/N$.  For an acquisition of $N$ complete windows, $T_{\rm tot}=NT$, and Chebyshev's inequality gives the following guarantee (see Appendix~\ref{app:upper-disjoint} for the full calculation).

\begin{theorem}[Learning from Bell measurements]\label{thm:upper}
Consider a stationary signal with the dependence-time bound $\tau_{\rm dep}$ of Eq.~\eqref{eq:mixsum}. Fix $\kappa>0$, finite squeezing $r\ge0$, and $0<\eps,\delta<1$. The estimator~\eqref{eq:estimator} estimates the characteristic function at every fixed $\abs{\beta}^2\le\kappa n$ to accuracy $\eps$ and success probability $1-\delta$ using the sufficient observation time
\begin{align}
 T_{\mathrm{tot},\rm EA}
 \coloneqq T\left\lceil\frac{e^{2\kappa n e^{-2r}}+\tau_{\rm dep}/T}{\delta\eps^2}\right\rceil.
 \label{eq:generalupper}
\end{align}
\end{theorem}

The sensor response $h$ enters this bound through the output window statistics and their dependence time. For fixed $\kappa$, logarithmic squeezing $r=\tfrac12\log n+O(1)$ keeps the readout contribution bounded. Theorem~\ref{thm:upper} then gives polynomial acquisition time whenever $\tau_{\rm dep}/T$ is bounded and $T$, $\eps^{-1}$, and $\delta^{-1}$ grow at most polynomially in $n$. This requires a mean photon number $\sinh^2r=O(n)$ in each probe and idler. The idler-free lower bounds below impose no bound on probe energy, so the exponential separation holds even when those protocols can use arbitrarily energetic probes.

\section{Limits of local idler-free sensing}\label{sec:lower}
We now ask whether local sensing without a protected idler can achieve the same learning efficiency. We allow arbitrary probe energy and consider both sequentially prepared probes and inputs entangled across the acquisition, with destructive readout in each time bin. For a common stationary signal family, we show that both classes require acquisition time exponential in the number of selected temporal modes, establishing the separation from entanglement-assisted sensing.

We prove the acquisition-time lower bounds by reducing learning to hypothesis testing~\cite{huang2022,oh} (see Appendix~\ref{app:lower-proof} for the full proof). The key idea is to construct stationary signals whose characteristic functions differ by more than the allowed estimation error. A successful learning scheme can then distinguish these signals, which requires a constant TVD between their measurement outcome distributions. On the other hand, we show that each pulse contributes only an exponentially small amount to this TVD after averaging over the hidden Fourier direction, even when the protocol uses classical feedback. Consequently, the protocol must encounter exponentially many pulses. Relating this number to the elapsed time gives the desired acquisition-time lower bound.

\subsection{Observation-time learning limit}
We first state the lower bound for sequentially prepared probes (see Appendix~\ref{sec:time} for explicit finite-time constants), then explain the stationary pulse construction and the proof.

\begin{theorem}[Observation time for sequential probe preparation]\label{thm:lower}
Fix even $n\ge8$, query scale $\kappa\ge15/16$, and accuracy $0<\eps\le1/(24n)$. Consider learning an unknown strictly stationary displacement signal with finite mean power and $\tau_{\rm dep}\le4e^2T$. Every sequentially prepared local idler-free protocol described in Sec.~\ref{sec:protocols} satisfying Eq.~\eqref{eq:learning} uniformly over these signal laws and every fixed allowed query, with $\delta=1/3$, requires acquisition time
\begin{align}
 \frac{T_{\mathrm{tot},\rm IF}}{T}
 =\Omega\left(\frac{(2\kappa+\sqrt{1+4\kappa^2})^{n/2}}{n^2\eps^2}\right).
 \label{eq:generallower}
\end{align}
\end{theorem}

\begin{figure}
\begin{tikzpicture}[
  x=1cm,y=1cm,
  font=\rmfamily\fontsize{8}{9.2}\selectfont,text=black,
  inner sep=1pt,
  line cap=round,line join=round,
  >={Stealth[length=1.5mm,width=1.0mm]},
  panel/.style={font=\rmfamily\fontsize{8}{9.2}\selectfont\bfseries,
    anchor=west},
  axis/.style={draw=black,line width=0.70pt},
  guide/.style={draw=black,line width=0.45pt,
    dash pattern=on 2pt off 1.8pt},
  note/.style={text=black},
  amplitude/.style={font=\rmfamily\fontsize{8}{9.2}\selectfont,
    inner sep=0pt},
  widthlabel/.style={inner sep=0pt}
]
\pgfmathsetlengthmacro{\pulsefigurexunit}{min(8.6cm,\linewidth)/8.6}
\pgfsetxvec{\pgfpoint{\pulsefigurexunit}{0pt}}
\path[use as bounding box] (0,-0.70) rectangle (8.60,6.95);

\def\panelleft{0.04}
\def\plotleft{0.35}
\def\plotright{8.25}
\def\startone{0.65}
\def\starttwo{3.80}
\def\startfinal{6.95}
\def\componentwidth{0.39375} 
\def\pulsebase{2.85}

\draw[draw=black,line width=0.55pt]
  (0.50,4.02)
  .. controls (0.50,4.10) and (0.58,4.14) .. (0.75,4.14)
  -- (6.73,4.14)
  .. controls (7.38,4.14) and (7.61,4.24) .. (7.685,4.72)
  .. controls (7.70,4.28) and (7.78,4.14) .. (7.94,4.14)
  .. controls (8.06,4.14) and (8.10,4.08) .. (8.10,4.02);

\node[panel] at (\panelleft,6.70) {(a) Random pulse arrivals};
\node[text=black] at (4.30,6.30)
  {$\Delta_e=\renewaloffset T+\operatorname{Exp}(1/T)$};
\draw[axis,->] (2.75,6.00) -- (5.25,6.00);
\draw[axis] (\plotleft,5.23) -- (6.20,5.23);
\draw[axis,->] (6.95,5.23) -- (\plotright,5.23);
\node[font=\rmfamily\fontsize{10}{11}\selectfont,text=black]
  at (6.575,5.23) {$\cdots$};
\node[note,anchor=west,text=black] at (8.33,5.23) {$t$};

\foreach \e/\x in {0/0.70,1/2.75,2/5.25,e/7.55}{
  \fill[probeBlue] (\x,5.23) rectangle (\x+0.27,5.83);
  \draw[axis] (\x,5.23) -- (\x,5.13);
  \node[text=black,fill=white] at (\x,4.94) {$\tau_{\e}$};
}
\draw[draw=black,line width=0.65pt,dash pattern=on 2pt off 1.8pt]
  (1.95,5.15) -- (1.95,5.90);
\node[text=black,fill=white] at (1.95,4.94) {$0$};

\node[panel] at (\panelleft,4.40)
  {(b) One pulse, expanded ($m=2$)};
\node[note,anchor=west] at (5.35,4.40)
  {$\operatorname{Re}\xi_e(t)$};

\foreach \x/\label in {
  \startone/{\tau_e},
  \starttwo/{\tau_e+\tau_B},
  \startfinal/{\tau_e+2\tau_B}}{
  \draw[guide] (\x,2.10) -- (\x,3.62);
  \node[anchor=north,text=black] at (\x,1.94) {$\label$};
}

\foreach \k/\x/\z/\a in {
  1/\startone/0.42/0.65,
  2/\starttwo/-0.32/0.49}{
  \fill[sensorTeal]
    (\x,\pulsebase) rectangle (\x+\componentwidth,\pulsebase+\z);
  \fill[signalOrange]
    (\x+\componentwidth,\pulsebase)
    rectangle (\x+2*\componentwidth,\pulsebase+\a);
  \fill[sensorTeal]
    (\x+2*\componentwidth,\pulsebase)
    rectangle (\x+3*\componentwidth,\pulsebase-\z);
}
\fill[preparePurple] (\startfinal,\pulsebase)
  rectangle (\startfinal+\componentwidth,\pulsebase+0.44);
\draw[axis,->] (\plotleft,\pulsebase) -- (\plotright,\pulsebase);
\node[note,anchor=west,text=black] at (8.33,\pulsebase) {$t$};

\def\amplituderow{3.88}
\foreach \k/\x in {1/\startone,2/\starttwo}{
  \node[amplitude,anchor=base,text=sensorTeal]
    at (\x+0.5*\componentwidth,\amplituderow) {$\backgroundamp{\k}$};
  \node[amplitude,anchor=base,text=signalOrange]
    at (\x+1.5*\componentwidth,\amplituderow) {$\signalamp{\k}$};
  \node[amplitude,anchor=base,text=sensorTeal]
    at (\x+2.5*\componentwidth,\amplituderow) {$-\backgroundamp{\k}$};
}
\node[amplitude,anchor=base,text=preparePurple]
  at (\startfinal+0.5*\componentwidth,\amplituderow) {$\finalamp$};
\node[note] at (2.65,2.48) {Zero field};

\node[panel] at (\panelleft,1.18) {(c) Component widths};

\def\widthunit{0.90}
\fill[sensorTeal] (\startone,0.45) rectangle (\startone+\widthunit,0.85);
\fill[signalOrange] (\startone+\widthunit,0.45)
  rectangle (\startone+2*\widthunit,0.85);
\fill[sensorTeal] (\startone+2*\widthunit,0.45)
  rectangle (\startone+3*\widthunit,0.85);
\fill[zeroGray] (\startone+3*\widthunit,0.45)
  rectangle (\startone+8*\widthunit,0.85);
\node[amplitude,text=white] at (\startone+0.5*\widthunit,0.65)
  {$\backgroundamp{k}$};
\node[amplitude,text=white] at (\startone+1.5*\widthunit,0.65)
  {$\signalamp{k}$};
\node[amplitude,text=white] at (\startone+2.5*\widthunit,0.65)
  {$-\backgroundamp{k}$};
\node[note] at (\startone+5.5*\widthunit,0.65) {Zero field};
\node[widthlabel,anchor=north] at (\startone+0.5*\widthunit,0.27) {$w$};
\node[widthlabel,anchor=north] at (\startone+1.5*\widthunit,0.27) {$w$};
\node[widthlabel,anchor=north] at (\startone+2.5*\widthunit,0.27) {$w$};
\node[widthlabel,anchor=north] at (\startone+5.5*\widthunit,0.27)
  {$\tau_B-3w$};
\node at (4.30,-0.48)
  {$w=\tau_B/8,\quad L_p=2\tau_B+w,\quad T=4\tau_B$};

\end{tikzpicture}

\caption{Stationary pulse sequence and the structure of one pulse, illustrated with $m=2$ signal components. (a) Pulse arrivals separated by $2T$ plus an exponential waiting time of rate $1/T$ (mean $T$); the origin represents a typical observation time. (b) The components of one pulse, with dashed guides marking offsets from its start. Group $k$ begins at $\tau_e+(k-1)\tc$ and contains $\zeta_k$, $\alpha_k$, and $-\zeta_k$, each of width $w=\tc/8$. The final background component $\zeta_F$ begins at $\tau_e+m\tc$ and has the same width. Heights schematically illustrate real parts of the complex random amplitudes. (c) The three components of width $w$ are followed by a zero-field interval of length $5w$, giving the ratio $1:1:1:5$ within one group period. The background law is common to all hypotheses; only the joint law of $\boldsymbol{\alpha}$ changes. The target window has duration $T=2m\tc$. The detector bin grid is fixed independently of the pulse arrivals.
}
\label{fig:stationary-pulses}
\end{figure}

\subsection{A hard stationary signal family}\label{sec:hard-family-main}
We construct the hard family directly in the output displacement field $\xi_\theta$, whose window statistics are the learning target. The simplest physical realization of this field uses $h(t)=\delta(t)$ and $S_\theta=\xi_\theta$ in Eq.~\eqref{eq:response}. Appendix~\ref{app:smooth-response} shows that smooth versions of the hard family can also be realized through a class of fixed, stable sensor responses with finite response time, including relaxation cascades and damped-oscillator responses~\cite{clerk2010}. It constructs stationary finite-power inputs satisfying $h*S_\theta=\xi_\theta$, while preserving the acquisition-time bounds.

 The construction must satisfy two conditions. First, the characteristic functions must remain distinguishable after averaging over random pulse positions. Second, a pulse crossing a measurement-bin boundary must not provide an effective idler.

To satisfy these conditions, consider a pulse starting at $\tau_e$ with $m\coloneqq n/2$ short signal components. Their complex amplitudes are $\alpha_1,\ldots,\alpha_m$, and successive components start $\tc$ apart. We place background components with amplitudes $\zeta_k$ and $-\zeta_k$ immediately before and after the $k$th signal component. Thus each group has the order $\zeta_k$, $\alpha_k$, $-\zeta_k$, with zero field between successive groups. We also add a final background component of amplitude $\zeta_F$ after the last group. All components have width $w\coloneqq \tc/8$ and the same integral normalization, giving a pulse duration $L_p\coloneqq m\tc+w<T$. Figure~\ref{fig:stationary-pulses}(b)--(c) illustrates the pulse groups and their component widths; Appendix~\ref{app:source} specifies the profiles and their normalization.

We distinguish the hypotheses by changing only the joint distribution of $\boldsymbol{\alpha}\coloneqq (\alpha_1,\ldots,\alpha_m)$. The background amplitudes have the same Gaussian distribution under all hypotheses and are independent of $\boldsymbol{\alpha}$. The role of the background can be understood by considering the following two cases.

\medskip
\noindent\textit{Preserving the signal feature.} When each group lies entirely within a measurement bin, the contributions from $\zeta_k$ and $-\zeta_k$ cancel, leaving the signal displacement $\alpha_k$. We show below that the resulting amplitude feature remains detectable after averaging over all pulse positions.

\medskip
\noindent\textit{Blocking an effective idler.} When one signal component crosses a bin boundary, all $m$ components cross corresponding boundaries because their starts are separated by $\tc$. The signal then occupies $m+1$ measured modes, although it has only $m$ complex signal amplitudes. The additional mode direction could therefore act as an effective idler. In this case, however, the paired background components fall in different bins and no longer cancel. We choose their covariance, together with $\zeta_F$, so that all $m+1$ modes receive independent Gaussian displacement noise of variance two. The final component supplies noise in the direction left unaffected by the $m$ pairs. The covariance and the resulting displacement map are given in Appendices~\ref{app:source} and \ref{sec:geometry}, respectively.

To distinguish the hypotheses, we use the same family of signal-amplitude distributions as in random displacement-channel learning~\cite{oh}. Under the null hypothesis, $\boldsymbol{\alpha}$ has the circular Gaussian density $p_0\coloneqq \CN(0,\sigma_{\alpha}^2\Id_m)$. For the query scale in Theorem~\ref{thm:lower}, we set
\begin{align}
 \sigma_{\alpha}^2=2\kappa+\sqrt{1+4\kappa^2}.
 \label{eq:concrete}
\end{align}
To define the alternatives, we draw a hidden Fourier direction $\eta\sim\CN(0,2\kappa\Id_m)$ once and keep it fixed throughout acquisition. We perturb $p_0$ along this direction with strength $\ep=6n\eps\le1/4$ and opposite signs:
\begin{align}
 p_{s\eta}(\boldsymbol{\alpha})\coloneqq p_0(\boldsymbol{\alpha})[1+4s\ep\sin(2\Im \boldsymbol{\alpha}^\dagger\eta)],\qquad s=\pm1.
 \label{eq:hardmarks}
\end{align}
These densities are nonnegative, normalized, and have the same mean-square amplitude. We label the signal law by $\theta=0$ or $(s,\eta)$, with both $s$ and $\eta$ fixed throughout a realization.

The parameter choices serve three purposes. Since the mean squared norm of $\eta$ is $\kappa n$, a constant fraction of directions lies within the allowed query ball.  The amplitude variance in Eq.~\eqref{eq:concrete} is chosen so that local idler-free measurements reveal exponentially little information about the hidden Fourier feature in each pulse (see Appendix~\ref{app:overlap}). The arrival positions that fully preserve the amplitude feature carry stationary weight $O(1/n)$, motivating a perturbation strength proportional to $n\eps$. Fourier transforming the perturbation gives two translated Gaussian peaks with opposite signs, in addition to the central Gaussian peak.  The translated peaks distinguish the hypotheses (see Appendix~\ref{sec:gap} for their explicit form).

To study acquisition from one continuing signal, we now embed these amplitude distributions in a stationary pulse process. We choose the pulse starts from a stationary renewal process with mean arrival rate $\rho=1/(3T)$.  The gaps between successive pulse starts consist of a fixed interval $2T$ followed by an exponential waiting time of mean $T$.  The origin is chosen at a typical time using the equilibrium renewal construction~\cite{thorisson1992} described in Appendix~\ref{app:source}. Figure~\ref{fig:stationary-pulses}(a) illustrates the random arrivals and the position of the observation origin. For each fixed hypothesis, we independently draw the signal and background amplitudes of each pulse from the distributions specified above and define $\xi_\theta(t)\coloneqq \sum_{e\in\mathbb Z}\xi_e(t)$, where the pulse starts $\{\tau_e\}_{e\in\mathbb Z}$ extend over the entire time axis.
Since the marked pulse sequence is invariant under all real time translations, the resulting field is strictly stationary. The minimum start spacing $2T$ ensures that a duration-$T$ window meets at most one pulse. Thus only the signal-amplitude law distinguishes the hypotheses; we next verify that its feature survives averaging over pulse positions.

\subsection{Hypothesis testing and the proof mechanism}\label{sec:proof-mechanism-main}
Following the learning-to-testing approach used in channel-learning lower bounds~\cite{chen2024tight,oh}, we show that accurate learning of this family requires a constant TVD between the null and the equal mixture of the signed alternatives. We then derive an upper bound on the same TVD for sequential local measurements. Comparing the two bounds proves Theorem~\ref{thm:lower}.

\medskip
\noindent\textit{A surviving signal feature.} For each hidden direction $\eta$, consider the query $\beta_\eta$ that places $\eta$ on the first $m$ coordinates of the target window and is zero elsewhere. When each complete pulse group lies in the corresponding bin, the background pairs cancel and this query recovers the amplitude feature.  For a set of directions with probability at least $9/20$, the query is allowed and the feature survives averaging over all positions, including empty windows and partial captures (see Appendix~\ref{app:window-testing}):
\begin{align}
 \abs{\chi_{T,s\eta}(\beta_\eta)-\chi_{T,0}(\beta_\eta)}\ge\frac94\eps>2\eps,\qquad s=\pm1.
 \label{eq:main-stationary-feature-gap}
\end{align}

\medskip
\noindent\textit{Learning requires distinguishable records.} The gap in Eq.~\eqref{eq:main-stationary-feature-gap} lets an estimate with error at most $\eps$ distinguish the null from either alternative: accept the null when the estimate is close to its known characteristic function, and reject it otherwise. The same rule distinguishes the null from an equal mixture of the signed alternatives, with the sign chosen once and held fixed throughout acquisition. We reveal $\eta$ after acquisition, so the test evaluates the required query using the common measurement record of Sec.~\ref{sec:setup}.

Let $P_\theta^{\rm tr}$ denote the distribution of this complete classical record, including the protocol's classical randomness. Since $\eta$ has the same distribution under both testing hypotheses,  the TVD of the direction and measurement record together is the average of the conditional TVD. On the separating directions, successful learning gives a test with success probability at least $2/3$; elsewhere, the test guesses randomly.  The probability of the separating directions therefore gives the testing requirement (see Appendix~\ref{app:window-testing})
\begin{align}
 \E_\eta\left[\TV\left(P_0^{\rm tr},\frac{P_{+\eta}^{\rm tr}+P_{-\eta}^{\rm tr}}2\right)\right]\ge\frac3{20}.
 \label{eq:main-testing-requirement}
\end{align}

\medskip
\noindent\textit{Limited response of one pulse.} We now upper-bound the same TVD for sequentially prepared probes. We may reveal the pulse schedule to the protocol: this can only increase the joint-data TVD, so an upper bound with this extra information also applies when the schedule is unknown.

The measurement outcomes respond only weakly to the hidden Fourier feature, even though it separates the target characteristic functions. For a fixed outcome sequence, this response is controlled by the characteristic function at $\eta$ of the pulse amplitudes conditioned on those outcomes under the null hypothesis. Its squared modulus, averaged over $\eta$, is at most $(\sigma_{\alpha}^2)^{-m}=\sigma_{\alpha}^{-n}$, uniformly over arrival times, probe states, and local measurements (see Appendix~\ref{app:overlap}). When each signal component lies within one bin, this follows from the displacement-overlap bound. When components cross bin boundaries, the background noise blocks the effective idler and gives the same bound. Thus neither probe choice nor pulse alignment removes the exponential suppression of this average response.

\medskip
\noindent\textit{From pulse responses to acquisition time.} For sequential preparation, fixing a complete outcome sequence also fixes all probe preparations and measurement choices along that sequence. The pulse likelihoods then factor because different pulses have independent amplitudes and occupy disjoint sets of bins. This allows us to bound their accumulated response even with classical feedback. For the equal mixture of the two fixed signs, the linear perturbations cancel, and the TVD is bounded by a sum of squared pulse responses. Averaging over outcomes, hidden directions, and pulse schedules gives (see Appendix~\ref{app:sequential})
\begin{align}
 \begin{aligned}
 \E_\eta\left[\TV\left(P_0^{\rm tr},\frac{P_{+\eta}^{\rm tr}+P_{-\eta}^{\rm tr}}2\right)\right]
 \le8\ep^2\sigma_{\alpha}^{-n}\frac{T_{\mathrm{tot}}+L_p}{3T}.
 \end{aligned}
 \label{eq:main-record-bound}
\end{align}
The bound separates the weak response of each pulse, through $\ep^2\sigma_{\alpha}^{-n}$, from the mean number of pulses encountered, $(T_{\mathrm{tot}}+L_p)/(3T)$. The latter includes pulses only partially measured at the acquisition boundaries. Since learning requires the constant TVD in Eq.~\eqref{eq:main-testing-requirement}, the exponential suppression per pulse must be compensated by exponentially many pulses. Their mean arrival rate is $1/(3T)$, so this requires exponentially long acquisition. Substituting $\ep=6n\eps$ and Eq.~\eqref{eq:concrete} gives Theorem~\ref{thm:lower}. 

\subsection{Temporally coherent inputs}\label{sec:coherent-main}
We now allow the input probes to be entangled across the entire acquisition, while keeping the local idler-free measurements of Sec.~\ref{sec:setup}.  We show that this more general class still requires exponentially long acquisition, although our lower bound has a smaller exponential rate (see Appendix~\ref{app:coherent-input} for the finite-time bound with explicit constants).

\begin{theorem}[Observation time with temporally entangled inputs]\label{thm:coherent-input}
For the signal class and parameter range of Theorem~\ref{thm:lower}, every local idler-free protocol of Sec.~\ref{sec:setup} with an arbitrary joint input state across acquisition, satisfying the same learning guarantee, requires acquisition time
\begin{align}
 \frac{T_{\mathrm{tot},\rm IF}}{T}=\Omega\left(\frac{\left[\frac{50}{81}\bigl(2\kappa+\sqrt{1+4\kappa^2}\bigr)\right]^{n/4}}{n^2\eps^2}\right).
 \label{eq:coherent-lower}
\end{align}
\end{theorem}
The bracketed factor exceeds one throughout the theorem's range.

\medskip

\noindent\textit{From scalar to operator responses.} Temporal input entanglement can correlate the outcomes of different pulses, so their scalar likelihoods need not factor. We instead represent each pulse's perturbation by an operator on its input modes, retaining the full joint input state. These operators act on disjoint mode blocks and commute, which allows the same signed-mixture argument to bound the TVD.

The resulting bound depends on operator second moments, which remain exponentially small in $n$. Combining this suppression with the same mean pulse count and testing requirement gives Theorem~\ref{thm:coherent-input}. The smaller exponent reflects the strength of our operator bound; it does not establish that temporal input entanglement achieves a corresponding improvement.

\subsection{Comparison on the same stationary signal}

We now compare the acquisition times for the same stationary signal family, whose finite mean power and dependence bound $\tau_{\rm dep}=O(T)$ are verified in Appendix~\ref{app:source-checks}. These properties allow us to apply Theorem~\ref{thm:upper} to the signals used in the lower bounds. For fixed $\kappa$, choose $r=\tfrac12\log(\kappa n)$ and use the accuracy and confidence of Theorem~\ref{thm:lower}. Then the Bell noise factor in the upper bound is $e^2$, and the mean photon number of each probe and idler is $\sinh^2r=O(n)$. Combining the bounds gives
\begin{align}
 \frac{T_{\mathrm{tot},\rm IF}}{T}
 &=\Omega\left(\frac{[\frac{50}{81}(2\kappa+\sqrt{1+4\kappa^2})]^{n/4}}{n^2\eps^2}\right),\\
 \frac{T_{\mathrm{tot},\rm EA}}{T}
 &=O\left(\frac1{\eps^2}\right).
 \label{eq:physical-scaling}
\end{align}
The idler-free lower bound allows arbitrary temporal input entanglement, while sequential preparation satisfies the stronger bound in Theorem~\ref{thm:lower}. In particular, for $\eps=\Theta(1/n)$ within the stated range, the entanglement-assisted scheme requires $O(n^2)$ windows of duration $T$, whereas local idler-free schemes require exponentially many. Therefore, we obtain a polynomial-versus-exponential separation in acquisition time for learning the same stationary window statistics. Here $n=T/\tc$ counts the selected temporal modes per target window: increasing $T$ at fixed readout rate includes longer temporal structure, while increasing $n$ at fixed $T$ compares finer interfaces with correspondingly rescaled hard signals.

For comparison, vacuum probes with heterodyne readout (VH) give the estimator and observation-time bound of Theorem~\ref{thm:upper} with $r=0$~\cite{oh}, corresponding to independent additive noise $\CN(0,1)$ in each bin. Figure~\ref{fig:acquisition-comparison} evaluates the finite-time bounds for these schemes and both classes of local idler-free protocols. It includes a fixed-squeezing example at $15\,\mathrm{dB}$, a level demonstrated for optical squeezed light~\cite{vahlbruch2016}, while retaining our ideal sensing assumptions. At fixed squeezing the EA upper bound still grows exponentially with $n$, albeit more slowly than the plotted IF lower bounds. Logarithmic squeezing instead gives the polynomial scaling above.

\begin{figure}[t]
 \centering
 \includegraphics[width=\columnwidth]{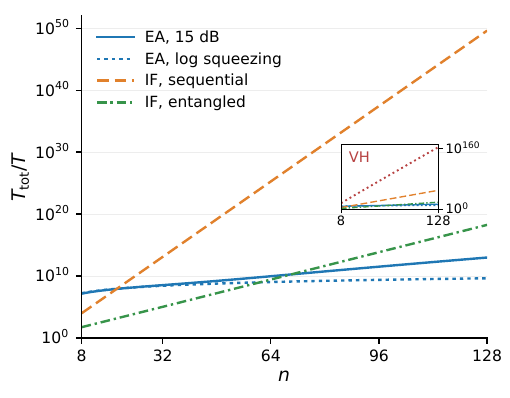}
 \caption{Finite-time acquisition bounds on the same stationary signal family, with $\kappa=1.4$, $\eps=1/(48n)$, and $\delta=1/3$. Blue curves give sufficient times for EA sensing from Theorem~\ref{thm:upper}, using $\tau_{\rm dep}/T=4e^2$: fixed $15\,\mathrm{dB}$ squeezing (solid) and $r=\tfrac12\log(\kappa n)$ (dotted), which reaches $22.5\,\mathrm{dB}$ at $n=128$. Orange and green curves give necessary times for sequential and temporally entangled IF protocols from Eqs.~\eqref{eq:sequential-finite} and \eqref{eq:coherent-finite}, using the finite-$n$ value of $q_m$ in Eq.~\eqref{eq:qgap-app} with $c=0.1$. The inset includes the VH sufficient-time bound ($r=0$). An IF lower bound above an EA upper bound certifies an advantage; crossings of the bounds need not coincide with the optimal performance crossover.}
 \label{fig:acquisition-comparison}
\end{figure}

The figure establishes a separation of rigorous bounds, but does not determine the optimal acquisition times. In particular, the large EA values are sufficient times and need not be necessary. The required accuracy $\eps=O(1/n)$ reflects the decreasing stationary weight of pulse positions that preserve the hidden feature, and contributes a factor $\eps^{-2}$ to the sufficient time. Thus the asymptotic advantage does not by itself imply a short acquisition in a finite experiment. Tighter finite-time bounds are needed to assess that regime.

\section{Conclusion}
We have established an exponential quantum advantage in the acquisition time required to learn stationary signal statistics from one continuously measured realization. Our entanglement-assisted estimator accounts for both readout noise and correlations between signal windows. For a common stationary pulse family, it achieves polynomial acquisition time with logarithmic squeezing and polynomial photon number, whereas local idler-free sensing requires exponentially long acquisition even with arbitrary probe energy, temporal input entanglement, and classical feedback.

The stationary setting connects exponential displacement-learning advantages to temporal statistics acquired from a continuing physical signal. The target distribution includes arbitrary pulse positions, partial captures, and empty intervals, while the acquisition cost includes waiting times and correlations between windows. No resets or event alignment are required. Because the learned coordinates are fixed by the sensor response and readout resolution, their joint statistics describe accessible correlations and fluctuations of the physical signal. This brings exponential learning advantages closer to continuous sensing tasks such as noise spectroscopy and the characterization of stationary fluctuations.

Our comparison identifies a protected idler as a resource that reduces acquisition time even when the probes may be entangled across time. The Bell scheme requires only local readout within each bin, and the idler can occupy the same temporal packet as the signal (see Appendix~\ref{app:internal-idler}). Quantum memories can enhance sensing sensitivity~\cite{zaiser2016}, and coherent processing can retrieve information missed by temporally local measurements~\cite{yang2023}. Extending the present learning comparison to such readout architectures, as well as to loss and other temporal-mode profiles, would clarify which resources are essential under experimental constraints.

Several quantitative questions remain open. Sharper finite-time bounds would determine whether the large sufficient acquisition times in Fig.~\ref{fig:acquisition-comparison} can be substantially reduced. It is also unknown whether temporal input entanglement changes the optimal exponential rate for local idler-free sensing. Finally, our construction establishes the separation at accuracy $\eps=O(1/n)$; obtaining an analogous stationary-window separation at an accuracy independent of $n$ would further strengthen its practical significance.

\bigskip

\noindent\textit{Note added.---}
While preparing this manuscript, we became aware of Ref.~\cite{mele2026continuous}, which also establishes exponential separations in total sensing time for continuously observed signals. That work studies squeezing-dependent separations for a single-mode Gaussian sensor. Our work establishes an exponential advantage from entanglement with protected idlers over local idler-free protocols, even with unlimited probe energy, non-Gaussian inputs and measurements, and temporally entangled probes under destructive single-bin readout.


\begin{acknowledgments}
    This work was supported by the National Research Foundation of Korea (NRF) grants funded by the Korean government (Ministry of Science and ICT, MSIT) (Nos. RS-2024-00431768 and RS-2025-00515456), and by the Institute of Information \& Communications Technology Planning \& Evaluation (IITP) grants funded by the Korean government (MSIT) (Nos. RS-2024-00437284, IITP-2025-RS-2025-02283189, and IITP-2025-RS-2025-02263264). This work was also supported by the Global Partnership Program of Leading Universities in Quantum Science and Technology (RS-2025-08542968) and the Creation of the Quantum Information Science R\&D Ecosystem Program (based on human resources, RS-2023-00256050), both through the NRF and funded by the Korean government (MSIT).
\par The authors used ChatGPT (OpenAI) to assist with detailed proof development, calculations, numerical code and figures, literature searches, and manuscript editing. The authors formulated the research problem and provided key conceptual ideas underlying the proofs, critically evaluated the AI-generated contributions, and independently verified the final arguments, numerical results, and references. The authors take full responsibility for the manuscript’s content and conclusions.\end{acknowledgments}

\bibliographystyle{apsrev4-2-etal}
\bibliography{main_references}

\clearpage
\appendix
\onecolumngrid
\addtocontents{toc}{\protect\fi}
\setcounter{tocdepth}{2}
\tableofcontents

\section{An effective idler in the same temporal packet}\label{app:internal-idler}
The sensing comparison in Sec.~\ref{sec:setup} depends on which modes the readout can access. We illustrate this with a complementary mode that acts as a protected idler within the same temporal packet and reproduces the Bell statistics in Eq.~\eqref{eq:bell}. The signal has a matched, constant profile within the packet. Let $u_1,u_2$ be normalized rectangular modes on its first and second halves. Their signal and complementary combinations are $\psi\coloneqq (u_1+u_2)/\sqrt2$ and $\phi\coloneqq (u_1-u_2)/\sqrt2$. Let $\alpha$ denote the displacement of $\psi$. The corresponding half-mode amplitudes are $\alpha/\sqrt2$, while $\phi$ is undisplaced. Prepare independent squeezed vacua with
\begin{align}
 \Var(x_{u_1})=\Var(p_{u_2})=e^{-2r},\qquad
 \Var(p_{u_1})=\Var(x_{u_2})=e^{2r}.
\end{align}
The signal gives $\E_{x_{u_1}}[x_{u_1}]=\sqrt2\Re\alpha$ and $\E_{p_{u_2}}[p_{u_2}]=\sqrt2\Im\alpha$. Their measured Gaussian noises are independent, so the normalized complex combination is
\begin{align}
 Y\coloneqq (x_{u_1}+ip_{u_2})/\sqrt2=\alpha+W,\qquad W\sim\CN(0,e^{-2r}).
 \label{eq:internal-idler}
\end{align} In the signal--complement basis, the same readout is
\begin{align}
 Y=\frac{x_\psi+x_\phi+i(p_\psi-p_\phi)}{2},\qquad
 [x_\psi+x_\phi,p_\psi-p_\phi]=0.
\end{align}
The variances of the two commuting EPR combinations are $2e^{-2r}$. To verify that the same input supplies probe--idler entanglement in this basis, use the covariance convention: for quadratures $\hat r_j$, write $V_{jk}\coloneqq \langle\{\Delta\hat r_j,\Delta\hat r_k\}\rangle/2$, with $\Delta\hat r_j\coloneqq \hat r_j-\langle\hat r_j\rangle$. The signal--complement covariance matrix has diagonal entries $\cosh(2r)$, with $V_{x_\psi x_\phi}=-\sinh(2r)$ and $V_{p_\psi p_\phi}=+\sinh(2r)$. This is the covariance matrix of a two-mode squeezed vacuum, up to a local phase convention. Thus product inputs in the half-mode basis can contain entanglement across the signal--complement mode decomposition. The orthogonal signal and complementary modes share the same temporal support.
The half-mode measurements may be performed in temporal order and their classical outcomes combined afterward. Their performance requires the specified signal profile and access to both half-modes; additional displacement of the complementary mode changes the readout. The two accessible half-modes thus realize a signal--idler readout within one temporal packet.

\section{Entanglement-assisted upper bound}\label{app:upper}\label{app:upper-disjoint}
We prove Theorem~\ref{thm:upper} for the estimator in Eq.~\eqref{eq:estimator}. The Bell measurement adds independent Gaussian noise in each bin, so the readout noise contributes only to the variance of each individual window, while the covariances between windows arise from the signal. Throughout this appendix, we write $\nu\coloneqq e^{-2r}$ for the Bell noise variance. The subscript $\theta$ in $\E_\theta$ and $\Pr_\theta$ means that the signal law is held fixed: the average and probability are over the signal realization and measurement outcomes, including the Bell noise. A subscript $\Xi$ denotes averaging only over the Bell noise.

For the TMSV phase used in Sec.~\ref{sec:upper}, the quadrature differences $\hat x_a-\hat x_i$ and $\hat p_a+\hat p_i$ commute and both have variance $2e^{-2r}$. The probe displacement shifts their means by $2\Re X_j$ and $2\Im X_j$. Dividing their complex combination by $2$ gives real and imaginary noise variances $\nu/2$, proving Eq.~\eqref{eq:bell}. Independent probe--idler pairs therefore give
\begin{align}
 \bm Y^{(i)}=\bm X_{iT}+\bm\Xi^{(i)},\qquad
 \bm\Xi^{(i)}\overset{\rm iid}{\sim}\CN(0,\nu\Id_n),
 \label{eq:upper-window-readout}
\end{align}
where the Bell noise vectors are independent of the entire signal. The signal vectors $\bm X_{iT}$ can remain correlated between windows.

We use the estimator in Eq.~\eqref{eq:estimator}, written as an average of its window summands $U_i$. We also define $f_i$ as the corresponding Fourier feature of the signal:
\begin{align}
 \hat\chi_{T,N}(\beta)=\frac1N\sum_{i=0}^{N-1}U_i,\qquad
 U_i\coloneqq e^{\nu\abs{\beta}^2}e^{(\bm Y^{(i)})^\dagger\beta-\beta^\dagger\bm Y^{(i)}},\qquad
 f_i\coloneqq e^{\bm X_{iT}^\dagger\beta-\beta^\dagger\bm X_{iT}}.
\end{align}
The Bell noise in Eq.~\eqref{eq:upper-window-readout} has Gaussian Fourier average
\begin{align}
 \E_\Xi[e^{(\bm\Xi^{(i)})^\dagger\beta-\beta^\dagger\bm\Xi^{(i)}}]=e^{-\nu\abs{\beta}^2}.
 \label{eq:conventions}
\end{align}
Since the Bell noise is independent of the signal, its average factors from the signal Fourier feature. The prefactor in $U_i$ cancels the Gaussian factor above, and stationarity gives
\begin{align}
 \E_\theta[U_i]=\E_\theta[f_i]=\chi_{T,\theta}(\beta),\qquad
 \E_\theta[\hat\chi_{T,N}(\beta)]=\chi_{T,\theta}(\beta).
 \label{eq:upper-unbiased}
\end{align}
The estimator is therefore unbiased, even when the signal windows are correlated.

We now calculate the variance of the estimator, using the complex covariance $\Cov(U,V)\coloneqq \E_\theta[U\overline V]-\E_\theta[U]\overline{\E_\theta[V]}$. Since $\abs{U_i}^2=e^{2\nu\abs{\beta}^2}$ and $\E_\theta[U_i]=\chi_{T,\theta}(\beta)$, we have
\begin{align}
 \Cov(U_i,U_i)=e^{2\nu\abs{\beta}^2}-\abs{\chi_{T,\theta}(\beta)}^2
 \le e^{2\nu\abs{\beta}^2}.
\end{align}
For $i\ne j$, the Bell noises are independent of one another and of the signal. Their Gaussian averages cancel both prefactors, giving $\E_\theta[U_i\overline U_j]=\E_\theta[f_i\overline f_j]$ and hence $\Cov(U_i,U_j)=\Cov(f_i,f_j)$.  Thus the off-diagonal covariances depend only on the signal. Since $f_i\overline f_j$ has unit modulus, the TVD between the joint signal-window law and the product of its marginals bounds
\begin{align}
 \abs{\Cov(U_i,U_j)}=\left|\E_\theta[f_i\overline f_j]-\E_\theta[f_i]\overline{\E_\theta[f_j]}\right|\le 2\beta_{\rm mix}(\abs{i-j}).
\end{align}
Expanding the variance of the average, and grouping the $N-g$ pairs at each separation $g$, now gives
\begin{align}
 \E_\theta[\abs{\hat\chi_{T,N}(\beta)-\chi_{T,\theta}(\beta)}^2]
 &=\Var_\theta(\hat\chi_{T,N}(\beta))
 =\frac1{N^2}\sum_{i,j=0}^{N-1}\Cov(U_i,U_j)\nonumber\\
 &\le\frac{e^{2\nu\abs{\beta}^2}}{N}
  +\frac{4}{N^2}\sum_{g=1}^{N-1}(N-g)\beta_{\rm mix}(g)
  \le\frac{e^{2\nu\kappa n}+\tau_{\rm dep}/T}{N}.
 \label{eq:variance-main}
\end{align}
The factor $4$ combines the two orientations of each pair with the factor $2$ in the TVD bound. The last line uses $\abs{\beta}^2\le\kappa n$ and the definition of $\tau_{\rm dep}$ in Eq.~\eqref{eq:mixsum}. The dependence time therefore bounds the accumulated covariance between windows.

Since the protocol observes $N$ consecutive windows, $T_{\rm tot}=NT$, and Chebyshev's inequality gives
\begin{align}
 \Pr_\theta\!\left[\abs{\hat\chi_{T,N}(\beta)-\chi_{T,\theta}(\beta)}>\eps\right]
 \le\frac{T e^{2\nu\kappa n}+\tau_{\rm dep}}{T_{\rm tot}\eps^2}.
 \label{eq:upper-confidence}
\end{align}
Choosing $N=\left\lceil(e^{2\nu\kappa n}+\tau_{\rm dep}/T)/(\delta\eps^2)\right\rceil$ makes this probability at most $\delta$ and proves Theorem~\ref{thm:upper}. In particular, the signal-correlation contribution to the mean-square error is at most $\tau_{\rm dep}/T_{\rm tot}$; squeezing reduces the readout term, while a longer acquisition reduces both terms.

\section{Proof of the observation-time lower bounds}\label{app:lower-proof}
In this appendix, we prove Theorems~\ref{thm:lower} and \ref{thm:coherent-input}. As in the main-text proof sketch, we reduce learning to hypothesis testing~\cite{huang2022,oh}. Accurate learning requires a constant TVD between the measurement outcome distributions, whereas we show that local sensing contributes only an exponentially small amount per pulse. Comparing these two bounds gives the required acquisition time. We organize the proof as follows.

\begin{enumerate}
 \item \textit{Construct the hard instance (Sec.~\ref{app:source}).} Define a strictly stationary null signal and signed alternatives, with a common pulse schedule and background. Only the signal-amplitude law carries the hidden feature.
 \item \textit{Reduce learning to testing (Sec.~\ref{app:window-testing}).} Show that the feature survives averaging over pulse positions. An accurate learner then distinguishes the null from the signed mixture, imposing a lower bound on the TVD of measurement outcomes.
 \item \textit{Control one pulse (Sec.~\ref{app:overlap}).} Bound its conditional Fourier response uniformly over arrival times, including pulses crossing bin boundaries. This supplies the scalar per-pulse estimate.
 \item \textit{Convert the TVD bound to acquisition time (Sec.~\ref{app:sequential}).} Insert the uniform pulse-response bound and average the pulse count, then compare with the testing requirement. This proves Theorem~\ref{thm:lower}.
 \item \textit{Include temporally entangled inputs (Sec.~\ref{app:coherent-input}).} Extend the likelihood-deficit calculation to correlated inputs and derive an operator second-moment bound from the scalar Fourier response. This proves Theorem~\ref{thm:coherent-input}.
 \item \textit{Verify the signal properties (Sec.~\ref{app:source-checks}).} Establish finite mean power and the dependence bound, placing the family in the stated signal class and enabling the Bell upper bound.
\end{enumerate}
Steps 1--5 derive the acquisition-time bounds on the hard family. The final step verifies the signal properties required by the theorem statements. We use $\TV(P,Q)\coloneqq \tfrac12\int|dP-dQ|$, for which optimal equal-prior testing success is $\tfrac12+\tfrac12\TV(P,Q)$.

\subsection{Setting and stationary signal family}\label{app:source}
We begin by constructing a stationary signal as a sum of pulses with random arrival times, $\xi_\theta(t)\coloneqq \sum_{e=-\infty}^{\infty}\xi_e(t)$. We specify the arrival process, the pulse shape, and the amplitude distributions in that order. Recall that the learning task in Sec.~\ref{sec:setup} concerns $n$ temporal modes within a window of duration $T=n\tc$, measured during a total time $\tobs$. The acquisition protocol is independent of the query, and for every fixed signal law and every fixed allowed query it must satisfy
\begin{align}
 \Prob_\theta[\abs{\hat\chi_T(\beta)-\chi_{T,\theta}(\beta)}\le\eps]\ge\frac23.
 \label{eq:lower-learning}
\end{align}
We treat the sequentially prepared protocols described in Sec.~\ref{sec:protocols} and the temporally entangled inputs of Theorem~\ref{thm:coherent-input}, both with destructive single-bin readout and causal classical feedback.

\medskip

Take even $n\ge8$ and set $m\coloneqq n/2$, the number of signal components in each pulse. The remaining window coordinates accommodate the final background component, keeping the pulse duration below $T$.

We construct stationary pulse arrivals using the standard equilibrium renewal construction~\cite{thorisson1992}. The ordinary gap between successive pulse starts is $2T$ plus an exponential waiting time of mean $T$, with density
\begin{align}
 \frac1T e^{-(d-2T)/T}\mathbf1_{d\ge2T}
\end{align}
and mean $3T$. Since each pulse lasts less than $T$, the minimum spacing ensures that a target window intersects at most one pulse.

To make the origin an arbitrary observation time, we sample its containing gap in proportion to its length: longer gaps contain more possible observation times. This gives the density $d e^{-(d-2T)/T}\mathbf1_{d\ge2T}/(3T^2)$. We then place the origin uniformly within this gap and extend the pulse starts in both directions using independent ordinary gaps. The resulting stationary process has mean pulse rate $\rho=1/(3T)$. We index the starts $\tau_e$ by $e\in\mathbb Z$ and independently draw each pulse's amplitudes from the common law specified below, independently of the arrival process.

\medskip
We next specify the waveform relative to each arrival time. Each pulse consists of $m$ rectangular groups with starts $s_{e,k}\coloneqq \tau_e+(k-1)\tc$, $k=1,\ldots,m$. Within each group, three intervals of width $w\coloneqq \tc/8$ carry the amplitudes $\zeta_k$, $\alpha_k$, and $-\zeta_k$ in that order. A final interval carries $\zeta_F$, as illustrated in Fig.~\ref{fig:stationary-pulses}(b)--(c). The complete waveform is
\begin{align}
 \xi_e(t)\coloneqq \frac{\sqrt\tc}{w}
 \begin{cases}
  \zeta_k, & s_{e,k}\le t<s_{e,k}+w,\\
  \alpha_k, & s_{e,k}+w\le t<s_{e,k}+2w,\\
  -\zeta_k, & s_{e,k}+2w\le t<s_{e,k}+3w,\\
  \zeta_F, & \tau_e+m\tc\le t<\tau_e+m\tc+w,\\
  0, & \text{otherwise}.
 \end{cases}
 \label{eq:event}
\end{align}
The first three cases apply for every $k=1,\ldots,m$. Since $3w<\tc$, the intervals are disjoint, with zero field between successive groups. The factor $\sqrt\tc/w$ makes a fully captured signal component contribute exactly $\alpha_k$ to the normalized time-bin displacement. The pulse is supported within $[\tau_e,\tau_e+L_p]$, with $L_p=m\tc+w<T$.

\medskip
We next define the amplitude distributions for the testing hypotheses. Let $\boldsymbol{\alpha}\coloneqq (\alpha_1,\ldots,\alpha_m)\in\mathbb C^m$ collect the signal amplitudes in one pulse. Under the null hypothesis, they follow a circular Gaussian distribution with variance $\sigma_{\alpha}^2>1$ per component; the explicit parameter range in the main text has $\sigma_{\alpha}^2\ge4$. Thus,
\begin{align}
 p_0(\boldsymbol{\alpha})\coloneqq \frac{e^{-\abs{\boldsymbol{\alpha}}^2/\sigma_{\alpha}^2}}{(\pi \sigma_{\alpha}^2)^m}.
\end{align}
To obtain different characteristic functions, we add a sinusoidal perturbation along a hidden Fourier direction $\eta$. This direction is drawn once and remains fixed throughout the signal realization. We choose its Gaussian distribution as
\begin{align}
 \eta\sim\CN(0,2\kappa\Id_m),\qquad \kappa>0.
 \label{eq:lower-parameters}
\end{align}
For each direction, the two alternatives add opposite sinusoidal perturbations of strength $0<\ep\le1/4$:
\begin{align}
 p_{s\eta}(\boldsymbol{\alpha})\coloneqq p_0(\boldsymbol{\alpha})[1+4s\ep\sin(2\Im \boldsymbol{\alpha}^\dagger\eta)],\qquad s=\pm1.
 \label{eq:lower-marks}
\end{align}
Thus $\E[\abs{\eta}^2]=\kappa n$, and a constant fraction of directions lies in the query ball. Section~\ref{app:window-testing} chooses the signal variance to obtain a constant fraction of separating queries, and Sec.~\ref{sec:time} sets the perturbation strength to match the requested accuracy.
The sine is odd under $\boldsymbol{\alpha}\mapsto-\boldsymbol{\alpha}$, while $p_0$ is even. Thus each density is normalized, and $4\ep\le1$ makes it nonnegative. The same symmetry leaves the even second moment $\E_{\boldsymbol{\alpha}\sim p_\theta}[\abs{\boldsymbol{\alpha}}^2]=m\sigma_{\alpha}^2$ unchanged. We use these densities to define two equally likely testing hypotheses.
\begin{enumerate}
 \item[(i)] \emph{Null hypothesis.} Every pulse has an independent signal vector drawn from $p_0$.
 \item[(ii)] \emph{Alternative hypothesis.} Draw $\eta$ from Eq.~\eqref{eq:lower-parameters} and choose $s=\pm1$ with equal probability, once for the entire acquisition. Conditional on this fixed pair $(s,\eta)$, every pulse has an independent signal vector drawn from $p_{s\eta}$.
\end{enumerate}
The schedule and Gaussian background have the same distribution under both hypotheses and are independent of the signal amplitudes. For each fixed $\eta$, hypothesis (ii) therefore gives an equal mixture of the two distributions of complete measurement outcomes. Note that the sign is chosen only once: averaging it independently for each pulse would give the null amplitude distribution. As in the main text, each stationary signal is indexed by $\theta=0$ or $(s,\eta)$. We prove stationarity and the dependence bound for each fixed $\theta$; the random choice of sign and direction is used only to define the hypothesis test.

When a signal component crosses a bin boundary, its preceding positive-noise component and following negative-noise component fall in adjacent bins. The captured background therefore consists of successive differences of the amplitudes. Write $\boldsymbol{\zeta}\coloneqq (\zeta_1,\ldots,\zeta_m,\zeta_F)^{\mathrm{T}}$ and let $\Delta_g$ take these differences. To fill every crossing mode with the independent noise required by Lemma~\ref{lem:noise}, we choose the background so that
\begin{align}
 \Delta_g\boldsymbol{\zeta}\coloneqq (\zeta_1,\zeta_2-\zeta_1,\ldots,\zeta_m-\zeta_{m-1},\zeta_F-\zeta_m)^{\mathrm{T}}\sim\CN(0,2\Id_{m+1}).
 \label{eq:noiseidentity}
\end{align}
We realize this by drawing $\zeta_1\sim\CN(0,2)$ and adding an independent $\CN(0,2)$ increment at each subsequent step, including the final amplitude $\zeta_F$. We draw all increments independently of $\boldsymbol{\alpha}$ and use the same law under every hypothesis. The resulting background vector has distribution
\begin{align}
 \boldsymbol{\zeta}\sim\CN(0,C_g),\qquad (C_g)_{jk}=2\min(j,k),\quad j,k=1,\ldots,m+1.
 \label{eq:noise}
\end{align}

Thus $\Delta_g C_g\Delta_g^{\mathrm{T}}=2\Id_{m+1}$, and every crossing mode receives independent displacement noise of variance two. Figure~\ref{fig: non-crossing and crossing case} illustrates background cancellation within bins and the noise produced at crossing boundaries; Sec.~\ref{sec:geometry} verifies these capture properties from the pulse profiles.

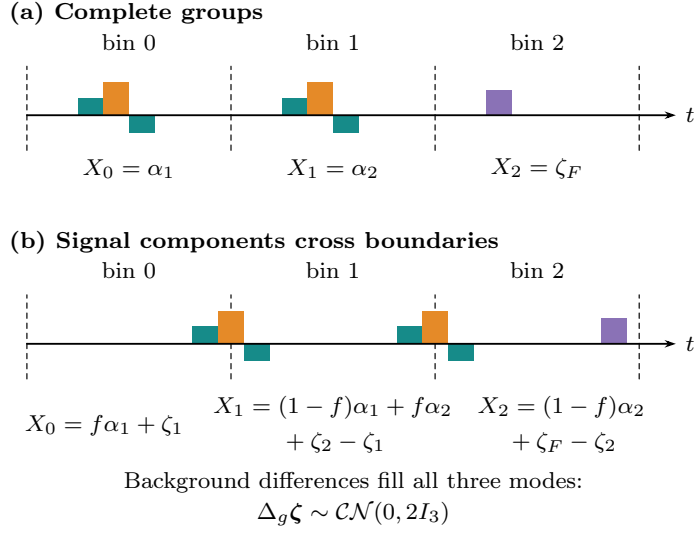
\begin{figure}
\centering
\begin{tikzpicture}[
    x=1cm,y=1cm,scale=1.08,transform shape,
  font=\rmfamily\fontsize{8}{9.2}\selectfont,text=black,
  inner sep=1pt,
  line cap=round,line join=round,
  >={Stealth[length=1.5mm,width=1.0mm]},
  panel/.style={font=\rmfamily\fontsize{8}{9.2}\selectfont\bfseries,
    anchor=west,text=black},
  axis/.style={draw=black,line width=0.70pt},
  guide/.style={draw=black,line width=0.45pt,
    dash pattern=on 2pt off 1.8pt},
  note/.style={align=center}
]
\pgfmathsetlengthmacro{\binfigurexunit}{min(8.6cm,\linewidth)/8.6}
\pgfsetxvec{\pgfpoint{\binfigurexunit}{0pt}}
\path[use as bounding box] (0,-0.35) rectangle (8.60,6.25);

\def\binleft{0.30}
\def\binperiod{2.50}
\def\bincomponentwidth{0.3125} 
\def\completegroupstart{0.925}
\def\boundaryshift{1.40625}

\node[panel] at (0.04,6.05) {(a) Complete groups};
\node[panel] at (0.04,3.25) {(b) Signal components cross boundaries};

\foreach \base/\offset in {4.80/0,2.00/\boundaryshift}{
  \foreach \j in {0,1,2,3}{
    \draw[guide]
      (\binleft+\j*\binperiod,\base-0.43)
      -- (\binleft+\j*\binperiod,\base+0.60);
  }
  \foreach \j in {0,1,2}{
    \node[text=black] at ({\binleft+(\j+0.5)*\binperiod},\base+0.90)
      {bin \j};
  }

  \foreach \k in {1,2}{
    \pgfmathsetmacro{\groupstart}{\completegroupstart+(\k-1)*\binperiod+\offset}
    \fill[sensorTeal]
      (\groupstart,\base)
      rectangle (\groupstart+\bincomponentwidth,\base+0.21);
    \fill[signalOrange]
      (\groupstart+\bincomponentwidth,\base)
      rectangle (\groupstart+2*\bincomponentwidth,\base+0.40);
    \fill[sensorTeal]
      (\groupstart+2*\bincomponentwidth,\base)
      rectangle (\groupstart+3*\bincomponentwidth,\base-0.21);

  }

  \pgfmathsetmacro{\finalstart}{\completegroupstart+2*\binperiod+\offset}
  \fill[preparePurple]
    (\finalstart,\base)
    rectangle (\finalstart+\bincomponentwidth,\base+0.31);

  \draw[axis,->] (\binleft,\base) -- (8.25,\base);
  \node[text=black,anchor=west] at (8.33,\base) {$t$};
}

\foreach \j/\result in {0/{\signalamp{1}},1/{\signalamp{2}},2/{\finalamp}}{
  \node[text=black] at ({\binleft+(\j+0.5)*\binperiod},4.15)
    {$X_{\j}=\result$};
}

\node[text=black] at (1.25,1.00)
  {$X_0=f\signalamp{1}+\backgroundamp{1}$};
\node[text=black] at (4.05,1.00)
  {$\begin{gathered}
    X_1=(1-f)\signalamp{1}+f\signalamp{2}\\
    {}+\backgroundamp{2}-\backgroundamp{1}
  \end{gathered}$};
\node[text=black] at (6.85,1.00)
  {$\begin{gathered}
    X_2=(1-f)\signalamp{2}\\
    {}+\finalamp-\backgroundamp{2}
  \end{gathered}$};

\node[note,text=black] at (4.30,0.30)
  {Background differences fill all three modes:};
\node[text=black] at (4.30,-0.08)
  {$\Delta_g\boldsymbol{\zeta}\sim\mathcal{CN}(0,2I_3)$};
\end{tikzpicture}
\caption{Pulse capture and background noise at bin boundaries, illustrated with two signal components. (a) Each complete group lies within a single bin, so the paired background contributions cancel, leaving the signal amplitude $\alpha_k$. The final background component occupies the following bin. The noncrossing requirement requires only that the signal component be in a single bin. Therefore, when noise is cut by bin boundaries, it does not cancel and instead appears as additive background displacements within the bins. 
(b) The signal components cross bin boundaries, and the background differences provide independent Gaussian displacement noise of variance two in all three captured modes. Here, $f$ is the fraction of each signal component captured in the earlier bin. The drawing shows $f=1/2$, while the displayed expressions hold for general $f$.}
\label{fig: non-crossing and crossing case}
\end{figure}

For each fixed $\theta$, the equilibrium renewal process is invariant under every real time translation~\cite{thorisson1992}. Attaching independent identically distributed amplitude vectors preserves this invariance, and the waveform map in Eq.~\eqref{eq:event} commutes with translations. The resulting field is locally square integrable and strictly stationary, so its retained windows have a common law. Section~\ref{app:source-checks} verifies the finite power and dependence bound required in Theorem~\ref{thm:lower}. We next establish a separating feature of the window law.

\subsection{Signal statistics and the learning-to-testing reduction}\label{app:window-testing}\label{sec:gap}\label{sec:testing-reduction}
We now show that a learner satisfying Eq.~\eqref{eq:lower-learning} can distinguish the two hypotheses. To do this, we first find a query for which the null and alternative characteristic functions differ by more than the estimation error. We then use the estimator to construct a hypothesis test and obtain a lower bound on the TVD of its measurement outcome distributions.

\medskip
\noindent\textit{Window characteristic function.} To find such a query, we calculate the window characteristic function by averaging over the pulse amplitudes and arrival times. By stationarity, we may take the target window to be $[0,T)$. For this fixed window, $\bm X_e\in\mathbb C^n$ records only the displacement captured from pulse $e$, including partial captures. It vanishes for starts outside $\mathcal R\coloneqq (-L_p,T)$. This interval has length less than $2T$, whereas successive starts are separated by at least $2T$, so at most one pulse contributes. Hence $\bm X=\sum_{e\in\mathbb Z}\bm X_e\in\mathbb C^n$ has at most one nonzero pulse contribution, and for every realization,
\begin{align}
 e^{\bm X^\dagger\beta-\beta^\dagger\bm X}
 =1+\sum_{e\in\mathbb Z}\left(e^{\bm X_e^\dagger\beta-\beta^\dagger\bm X_e}-1\right).
\end{align}
The baseline $1$ is the Fourier factor of an empty window. Each summand is its change due to one pulse and is zero for pulses outside the window.

We now express the captured contribution of a pulse in terms of its amplitudes. Consider pulse $e$ starting at time $\tau$ relative to the target window $[0,T)$. Depending on $\tau$, each signal component can lie inside one bin, cross a bin boundary, or extend outside the window. We write the resulting displacement as $\bm X_e=M_\tau\boldsymbol{\alpha}+N_\tau\boldsymbol{\zeta}$, where the two matrices describe the captured signal and background components. Specifically, $M_\tau$ is a real $n\times m$ matrix whose rows correspond to window bins $j=0,\ldots,n-1$ and columns to signal components $k=1,\ldots,m$. Its entry $(M_\tau)_{jk}$ is the fraction of component $k$ captured in bin $j$:
\begin{align}
 (M_\tau)_{jk}\coloneqq \frac{\bigl|[j\tc,(j+1)\tc)\cap[\tau+(k-1)\tc+w,\tau+(k-1)\tc+2w)\bigr|}{w}.
\end{align}
Here $|I|$ denotes the length of an interval. The background-capture matrix $N_\tau$ is a real $n\times(m+1)$ matrix. For each paired background amplitude $\zeta_k$, its entry is the captured fraction of the positive interval minus that of the negative interval. The last column records the captured fraction of $\zeta_F$.   Using the group starts $s_{e,k}$ with $\tau_e=\tau$, the paired columns ($k=1,\ldots,m$) and final column have entries
\begin{align}
   (N_\tau)_{jk}&\coloneqq \frac{\bigl|[j\tc,(j+1)\tc)\cap[s_{e,k},s_{e,k}+w)\bigr|-\bigl|[j\tc,(j+1)\tc)\cap[s_{e,k}+2w,s_{e,k}+3w)\bigr|}{w},
 \label{eq:background-capture-pairs}\\
 (N_\tau)_{j,m+1}&\coloneqq \frac{\bigl|[j\tc,(j+1)\tc)\cap[\tau+m\tc,\tau+m\tc+w)\bigr|}{w}.
 \label{eq:background-capture-final}
\end{align}
These entries follow from the background intervals in Eq.~\eqref{eq:event}. The capture matrices include both splitting at bin boundaries and truncation at the window edges.

We first average over the amplitudes with the pulse start $\tau$ fixed. Substituting $\bm X_e=M_\tau\boldsymbol{\alpha}+N_\tau\boldsymbol{\zeta}$ separates the Fourier exponent into signal and background terms. Since the amplitudes are independent, this average factorizes. Write $\lambda_\theta(z)\coloneqq \mathbb E_{\boldsymbol{\alpha}\sim p_\theta}[e^{\boldsymbol{\alpha}^\dagger z-z^\dagger\boldsymbol{\alpha}}]$ for the signal-amplitude characteristic function. The common Gaussian background has covariance $C_g$, giving
\begin{align}
 \mathbb E_{\boldsymbol{\alpha},\boldsymbol{\zeta}}[e^{\bm X_e^\dagger\beta-\beta^\dagger\bm X_e}]=\lambda_\theta(M_\tau^{\mathrm{T}}\beta)\Lambda_\tau(\beta),\qquad \Lambda_\tau(\beta)\coloneqq e^{-\beta^\dagger N_\tau C_gN_\tau^{\mathrm{T}}\beta}.
 \label{eq:fixed-arrival-fourier}
\end{align}
Here $M_\tau^{\mathrm{T}}\beta$ is the induced query on the signal amplitudes, and $\Lambda_\tau(\beta)$ is the background attenuation. Thus the amplitude-averaged change from one pulse is $\lambda_\theta(M_\tau^{\mathrm{T}}\beta)\Lambda_\tau(\beta)-1$.

We next average over pulse positions. For a stationary arrival process with rate $\rho$, the expected number of starts in a small interval $d\tau$ is $\rho d\tau$. Averaging a sum of pulse contributions therefore amounts to integrating the contribution at each possible start time against $\rho d\tau$. This is Campbell's formula~\cite{kallenberg}: for any integrable function $F$,
\begin{align}
 \mathbb E_{\{\tau_e\}}\!\left[\sum_{e\in\mathbb Z}F(\tau_e)\right]
 =\rho\int_{\mathbb R}F(\tau)d\tau.
 \label{eq:campbell-first-moment}
\end{align}
Applying this formula to the amplitude-averaged changes in the preceding sum gives
\begin{align}
 \chi_{T,\theta}(\beta)=1+\rho\int_{-L_p}^{T}
 \left\{\lambda_\theta(M_\tau^{\mathrm{T}}\beta)\Lambda_\tau(\beta)-1\right\}d\tau.
 \label{eq:campbell}
\end{align}
The integral is restricted to $(-L_p,T)$ because starts outside this interval contribute zero. It averages over all pulse positions relative to the fixed window, including complete and partial captures.

\medskip
\noindent\textit{Preserving the amplitude feature.} We first identify a feature that distinguishes the pulse-amplitude laws, then choose a window query that reads this feature when the pulse groups are fully captured. Fourier transforming the amplitude densities in Eq.~\eqref{eq:lower-marks} gives
\begin{align}
 \lambda_0(z)=e^{-\sigma_{\alpha}^2\abs{z}^2},\qquad
 \lambda_{s\eta}(z)-\lambda_0(z)
 =2is\ep\{\lambda_0(z-\eta)-\lambda_0(z+\eta)\}. \label{eq:markCF}
\end{align}
At $z=\eta$, the first Gaussian equals one and the second equals $e^{-4\sigma_\alpha^2|\eta|^2}$. We use this amplitude difference to distinguish the window laws.

To read this difference from the window, consider a pulse starting at $0<\tau<\tc-3w$. Its first group, of duration $3w$, lies wholly inside the first bin. Since the groups are spaced by $\tc$, all $m$ groups lie wholly inside the first $m$ bins. In each bin the $\zeta_k$ and $-\zeta_k$ contributions cancel, leaving $\alpha_k$. The final background component lies in the following bin. We therefore put $\eta$ on the first $m$ window coordinates and zero on the rest:
\begin{align}
 \beta_\eta\coloneqq (\eta_1,\ldots,\eta_m,0,\ldots,0),\qquad \abs{\beta_\eta}^2=\abs{\eta}^2.
 \label{eq:separating-query}
\end{align}
For these pulse positions, the window query reads exactly the amplitude feature:
\begin{align}
 \bm X_e^\dagger\beta_\eta-\beta_\eta^\dagger\bm X_e
 =\boldsymbol{\alpha}^\dagger\eta-\eta^\dagger\boldsymbol{\alpha}.
\end{align}
The query is fixed independently of the pulse positions; the equality holds throughout the arrival interval $\mathcal C\coloneqq (0,\tc-3w)$, whose length is
\begin{align}
 \lc\coloneqq \tc-3w=5\tc/8.
 \label{eq:clean}
\end{align}

To express this capture property in the notation used below, let $\mathcal J\coloneqq \{0,\ldots,m-1\}$ denote the first $m$ bins and $M_{\tau,\mathcal J}$ the corresponding rows of $M_\tau$. For any arrival time, the induced amplitude query is $M_\tau^{\mathrm T}\beta_\eta=M_{\tau,\mathcal J}^{\mathrm T}\eta$. On the interval $\mathcal C$, full capture and background cancellation give
\begin{align}
 M_{\tau,\mathcal J}=\Id,\qquad N_\tau^{\mathrm{T}}\beta_\eta=0,\qquad \Lambda_\tau(\beta_\eta)=1\quad(\tau\in\mathcal C).
 \label{eq:clean-capture}
\end{align}
Consequently, the positive Gaussian peak in Eq.~\eqref{eq:markCF} contributes one throughout $\mathcal C$, giving an arrival integral of at least $\lc$. We next bound the subtracted peak after averaging over all other pulse positions.

\medskip
\noindent\textit{Separation after averaging over arrivals.} Other pulse positions can only add to the nonnegative peak. We retain the positive contribution from complete groups and bound the subtracted peak. With $B_\tau\coloneqq \Id+M_{\tau,\mathcal J}^{\mathrm T}$, Eqs.~\eqref{eq:campbell}--\eqref{eq:markCF} give
\begin{align}
 \chi_{T,s\eta}(\beta_\eta)-\chi_{T,0}(\beta_\eta)
 &=2is\ep\rho\int_{\mathcal R}\Lambda_\tau(\beta_\eta)
 \left[e^{-\sigma_{\alpha}^2\abs{(M_{\tau,\mathcal J}^{\mathrm{T}}-\Id)\eta}^2}
 -e^{-\sigma_{\alpha}^2\abs{B_\tau\eta}^2}\right]d\tau,\\
 |\chi_{T,s\eta}(\beta_\eta)-\chi_{T,0}(\beta_\eta)|
 &\ge2\ep\rho[\lc-\mathcal E_-(\eta)],
 \label{eq:gapintegral}\end{align}
where $\mathcal E_-(\eta)\coloneqq \int_{\mathcal R}e^{-\sigma_{\alpha}^2\abs{B_\tau\eta}^2}d\tau$ bounds the total subtracted contribution: we have omitted the background attenuation using $0<\Lambda_\tau\le1$. The right-hand side is the preserved contribution minus an upper bound on this loss. Thus a detectable gap remains whenever $\mathcal E_-(\eta)$ is sufficiently small compared with $\lc$. Fix $0<c<1$ and collect the directions for which the loss is at most $c\lc$ and the query is allowed:
\begin{align}
 \mathcal G\coloneqq \{\eta:\abs{\eta}^2\le\kappa n,\quad\mathcal E_-(\eta)\le c\lc\}.
 \label{eq:good-query-directions}
\end{align} For every $\eta\in\mathcal G$, the loss is at most $c\lc$, so Eq.~\eqref{eq:gapintegral} ensures a positive feature gap. The factor $\rho\lc=\lc/(3T)$ is the stationary weight of the complete-capture interval and sets the inverse-$n$ feature scale.

\medskip
\noindent\textit{Probability of retaining the gap.}   We now lower-bound $\Prob_\eta(\mathcal G)$, the probability that a Gaussian direction gives an allowed query with a surviving feature gap. By the definition of $\mathcal G$,
\begin{align}
 \Prob_\eta(\mathcal G)\ge\Prob_\eta[\abs{\eta}^2\le\kappa n]-\Prob_\eta[\mathcal E_-(\eta)>c\lc].
\end{align}
For the first term, $\abs{\eta}^2/(2\kappa)$ is a sum of $m$ independent unit-mean exponential variables. Writing its distribution function as $F_m(x)\coloneqq1-e^{-x}\sum_{j=0}^{m-1}x^j/j!$, we obtain $\Prob_\eta[\abs{\eta}^2\le\kappa n]=F_m(n/2)$.

To bound the second term, we calculate the mean loss.
Index the selected bins and signal components locally by $i,k=1,\ldots,m$, and let $(S_d)_{ik}\coloneqq \mathbf1_{i-k=d}$. The equal spacing between signal components gives $M_{\tau,\mathcal J}=fS_d+(1-f)S_{d+1}$ for some integer $d$ and $f\in[0,1]$, including truncated captures at window boundaries. The matrix $B_\tau$ is triangular. Its diagonal is $1+f$ for $d=0$, $2-f$ for $d=-1$, and $1$ otherwise. In particular $|\det B_\tau|\ge1$. Gaussian integration therefore gives
\begin{align}
 \E_\eta[e^{-\sigma_{\alpha}^2\abs{B_\tau\eta}^2}]
 =\det(\Id+2\kappa\sigma_{\alpha}^2 B_\tau^\dagger B_\tau)^{-1}\le(1+2\kappa\sigma_{\alpha}^2)^{-m}. \label{eq:detgap}
\end{align}
To show the last inequality, let the positive eigenvalues of $B_\tau^\dagger B_\tau$ be $e^{t_k}$. The function $t\mapsto\log(1+2\kappa\sigma_{\alpha}^2 e^t)$ is increasing and convex, and $\sum_{k=1}^{m}t_k=\log|\det B_\tau|^2\ge0$.  Jensen's inequality proves the claim. Integrating over the arrival range of length $T+L_p$ gives
\begin{align}
 \E_\eta[\mathcal E_-(\eta)]\le(T+L_p)(1+2\kappa\sigma_{\alpha}^2)^{-m}.
\end{align}
Markov's inequality then gives
\begin{align}
 \Prob_\eta[\mathcal E_-(\eta)>c\lc]\le\frac{\E_\eta[\mathcal E_-(\eta)]}{c\lc}\le\frac{T+L_p}{c\lc}(1+2\kappa\sigma_{\alpha}^2)^{-m}.
\end{align}
Substituting the bounds for both terms yields the following result.

\medskip
\begin{proposition}[Separation of the signal features]\label{prop:gap}
Every direction $\eta\in\mathcal G$ gives an allowed query $\abs{\beta_\eta}^2\le\kappa n$ with
\begin{align}
 |\chi_{T,s\eta}(\beta_\eta)-\chi_{T,0}(\beta_\eta)|\ge 2\ep\rho\lc(1-c)\eqqcolon\gamma\qquad(s=\pm1).
 \label{eq:featuregap}
\end{align}
Moreover,
\begin{align}
 \Prob_\eta(\mathcal G)\ge F_m\!\left(\frac n2\right)-\frac{T+L_p}{c\lc}(1+2\kappa\sigma_{\alpha}^2)^{-m}\eqqcolon q_m.
 \label{eq:qgap-app}
\end{align}
\end{proposition}
\medskip
\noindent\textit{A constant fraction of separating queries.} We now evaluate the probability bound in Proposition~\ref{prop:gap} for the parameters used in the main text. Take
\begin{align}
 m=n/2,\qquad \sigma_{\alpha}^2=2\kappa+\sqrt{1+4\kappa^2},\qquad c=\frac1{10}.
 \label{eq:lower-scale-choice}
\end{align}
The mean squared hidden-direction norm equals the query-ball radius. For $m\ge4$ and $\kappa\ge15/16$, Eq.~\eqref{eq:qgap-app} gives
\begin{align}
 q_m\ge\frac12-(48m+2)(17/2)^{-m}>\frac9{20}.
 \label{eq:main-query-probability}
\end{align}
Here $F_m(m)\ge1/2$: in its Poisson representation $F_m(m)=\Prob[\operatorname{Pois}(m)\ge m]$, the probability at $m+j$ is at least that at $m-1-j$ for $0\le j<m$, as follows by pairing factors around $m$ in the factorial ratio. Summing these pairs proves the claim. Also $2\kappa\sigma_{\alpha}^2\ge15/2$, and the subtracted term is largest at $m=4$, where it is less than $0.038$. Thus $\Prob_\eta(\mathcal G)>9/20$: at least 45\% of directions give an allowed query whose feature gap survives averaging over all pulse positions. We use this constant fraction to convert a learning guarantee into a testing requirement.

\medskip
\noindent\textit{From learning to testing.} Suppose a protocol satisfies Eq.~\eqref{eq:lower-learning} at accuracy $\eps<\gamma/2$. We use its estimate to distinguish the null signal from the equal mixture of the two signed alternatives. A referee draws $\eta\sim\CN(0,2\kappa\Id_m)$ and chooses either hypothesis with equal probability. Under the alternative, the sign $s=\pm1$ is sampled uniformly once and held fixed throughout acquisition. The referee reveals $\eta$ only after acquisition. Write $P_\theta^{\rm tr}$ for the distribution of the complete classical record during $[0,\tobs)$, including the protocol's classical randomness.

For a revealed direction $\eta\in\mathcal G$, we evaluate the query $\beta_\eta$ and accept the null if
\begin{align}
 |\hat\chi_T-\chi_{T,0}(\beta_\eta)|<\gamma/2.
\end{align}
For $\eta\notin\mathcal G$, we flip a fair coin. To see why this test succeeds, note that when the estimate is accurate, its distance from the null value is at most $\eps<\gamma/2$ under the null and at least $\gamma-\eps>\gamma/2$ under either alternative. Thus, for every $\eta\in\mathcal G$, the test succeeds with probability at least $2/3$ under either hypothesis. Since the test succeeds with probability $1/2$ outside $\mathcal G$ and $\Prob_\eta(\mathcal G)\ge q_m$, its overall success probability is at least $1/2+q_m/6$.

Since the referee reveals $\eta$ after acquisition, the test receives the pair $(\eta,o)$, where $o$ is the complete outcome sequence. Let $g_\kappa(\eta)$ be the Gaussian density in Eq.~\eqref{eq:lower-parameters}. Writing outcome probabilities for discrete outcomes, the joint data distributions under hypotheses (i) and (ii) are
\begin{align}
 \mathsf P_0(\eta,o)\coloneqq g_\kappa(\eta)P_0^{\rm tr}(o),\qquad
 \mathsf P_1(\eta,o)\coloneqq g_\kappa(\eta)\frac{P_{+\eta}^{\rm tr}(o)+P_{-\eta}^{\rm tr}(o)}2.
 \label{eq:testing-joint-laws}
\end{align}
These expressions are densities in $\eta$ and probabilities in $o$; the corresponding conditional-distribution identity also holds for continuous readouts, treated in Sec.~\ref{app:outcomes}. Both hypotheses have the same direction density $g_\kappa$, so integrating the absolute difference over $\eta$ and summing over $o$ gives the joint-data TVD below. Since equal-prior binary testing has optimal success probability $\tfrac12+\tfrac12\TV(\mathsf P_0,\mathsf P_1)$, the preceding test implies
\begin{align}
 \TV(\mathsf P_0,\mathsf P_1)
 =\E_\eta\left[\TV\left(P_0^{\rm tr},\frac{P_{+\eta}^{\rm tr}+P_{-\eta}^{\rm tr}}2\right)\right]
 \ge\frac{q_m}{3}.
 \label{eq:TVlower}
\end{align}
For each fixed $(s,\eta)$, the test uses only the query $\beta_\eta$, independently of the realized outcomes. The pointwise guarantee in Eq.~\eqref{eq:learning} therefore suffices.

Equation~\eqref{eq:TVlower} completes the reduction: learning requires an average TVD of at least $q_m/3$. We next bound how much each encountered pulse can contribute to this same quantity under local readout.

\subsection{From record distinguishability to single-pulse responses}\label{app:overlap}\label{sec:geometry}
We now upper-bound the TVD achievable by sequentially prepared local idler-free probes and compare it with the learning requirement in Eq.~\eqref{eq:TVlower}. We first express the TVD through the alternative-to-null likelihood ratio. For a fixed hidden direction $\eta$ and complete outcome sequence $o$,
\begin{align}
 \TV\left(P_0^{\rm tr},\frac{P_{+\eta}^{\rm tr}+P_{-\eta}^{\rm tr}}2\right)
 =\frac12\E_{o\sim P_0^{\rm tr}}\left[\left|\frac{P_{+\eta}^{\rm tr}(o)+P_{-\eta}^{\rm tr}(o)}{2P_0^{\rm tr}(o)}-1\right|\right].
 \label{eq:record-tvd-likelihood}
\end{align}
Write the signed-mixture likelihood ratio as $L_\eta(o)\coloneqq [P_{+\eta}^{\rm tr}(o)+P_{-\eta}^{\rm tr}(o)]/[2P_0^{\rm tr}(o)]$. Normalization gives $\E_{o\sim P_0^{\rm tr}}[L_\eta(o)]=1$ and hence $\E_{o\sim P_0^{\rm tr}}[L_\eta(o)-1]=0$. The positive and negative parts therefore have equal null averages, giving, with $(x)_+\coloneqq \max(x,0)$,
\begin{align}
 \TV\!\left(P_0^{\rm tr},\frac{P_{+\eta}^{\rm tr}+P_{-\eta}^{\rm tr}}2\right)
 =\frac12\E_{o\sim P_0^{\rm tr}}[|L_\eta(o)-1|]
 =\E_{o\sim P_0^{\rm tr}}[(1-L_\eta(o))_+].
 \label{eq:record-positive-deficit}
\end{align}
Our goal is therefore a lower bound on the full-record ratio $L_\eta(o)$. We derive this bound first for sequential preparation, following Ref.~\cite{oh}, Supplemental Material, Sec. S3 A. We allow the protocol to know the pulse schedule: the schedule has the same law under both hypotheses, and revealing it can only increase the TVD. Fix the schedule and classical randomness, and let $M$ count the pulses met during acquisition. Below, $L_\eta(o\mid\mathrm{schedule})$ denotes the corresponding conditional ratio, for which the same positive-deficit identity holds. We average over the random schedule and classical randomness in Sec.~\ref{app:sequential}.

For a fixed complete outcome sequence $o$, all causal probe and readout choices are fixed. Write $o_{<j}\coloneqq (o_0,\ldots,o_{j-1})$ for the preceding outcomes, $\rho_j(o_{<j})$ for the prepared probe state, and $E_j(o_j\mid o_{<j})$ for the local outcome effect. Their outcome-history dependence includes classical feedback. Let $B_e$ contain the observed bins met by pulse $e$. The end-to-start gap between successive pulses is at least $2T-L_p>\tc$, so the sets $B_e$ are disjoint. In a bin $j\in B_e$, the captured displacement $X_j(\boldsymbol{\alpha}_e,\boldsymbol{\zeta}_e)$ therefore depends only on the signal and background amplitudes of pulse $e$.

If all pulse amplitudes are also fixed, sequential preparation and local readout give the record probability by the chain rule and the binwise Born rule:
\begin{align}
 P^{\rm tr}\left(o\mid\mathrm{schedule},\{\boldsymbol{\alpha}_e,\boldsymbol{\zeta}_e\}_{e=1}^{M}\right)
 =C(o)\prod_{e=1}^{M}\prod_{j\in B_e}
 \Tr\{E_j(o_j\mid o_{<j})\hat D(X_j(\boldsymbol{\alpha}_e,\boldsymbol{\zeta}_e))\rho_j(o_{<j})\hat D^\dagger(X_j(\boldsymbol{\alpha}_e,\boldsymbol{\zeta}_e))\}.
 \label{eq:record-fixed-amplitudes}
\end{align}
Here $C(o)\coloneqq\prod_{j\notin\bigcup_e B_e}\Tr[E_j(o_j\mid o_{<j})\rho_j(o_{<j})]$ contains the bins without a pulse. For this fixed branch, it is independent of the pulse amplitudes and the signal hypothesis. Although the outcomes can be correlated by feedback, fixing $o$ fixes every state and effect in this expression; no independence of the outcomes is assumed.

For each pulse, average its binwise product over its background amplitudes and define
\begin{align}
 \ell_e(\boldsymbol{\alpha};o)\coloneqq \E_{\boldsymbol{\zeta}}\left[\prod_{j\in B_e}
 \Tr\{E_j(o_j\mid o_{<j})\hat D(X_j(\boldsymbol{\alpha},\boldsymbol{\zeta}))\rho_j(o_{<j})\hat D^\dagger(X_j(\boldsymbol{\alpha},\boldsymbol{\zeta}))\}\right]. \label{eq:eventlik}
\end{align}
The background average is outside the product because the same background amplitudes contribute to several bins of one pulse. For pulses partially captured at the acquisition boundaries, we complete the unobserved bins with vacuum inputs and identity effects; each added Born factor is one, so the likelihood is unchanged.

Under a fixed signal hypothesis $\theta$, the signal and background amplitudes are independent between pulses. Averaging Eq.~\eqref{eq:record-fixed-amplitudes} first over the backgrounds and then over the signal amplitudes therefore gives
\begin{align}
 P_\theta^{\rm tr}(o\mid\mathrm{schedule})
 =C(o)\E_{\boldsymbol{\alpha}_1,\ldots,\boldsymbol{\alpha}_M\sim p_\theta^{\otimes M}}\left[\prod_{e=1}^{M}\ell_e(\boldsymbol{\alpha}_e;o)\right]
 =C(o)\prod_{e=1}^{M}\E_{\boldsymbol{\alpha}\sim p_\theta}[\ell_e(\boldsymbol{\alpha};o)].
 \label{eq:record-pulse-factorization}
\end{align}
Dividing the expression for $\theta=s\eta$ by that for $\theta=0$ cancels $C(o)$ and leaves the product of the single-pulse ratios. The sign $s$ is chosen once for the entire acquisition, so averaging the two signs is performed outside this product:
\begin{align}
 L_\eta(o\mid\mathrm{schedule})
 =\frac12\sum_{s=\pm1}\prod_{e=1}^{M}
 \frac{\E_{\boldsymbol{\alpha}\sim p_{s\eta}}[\ell_e(\boldsymbol{\alpha};o)]}
 {\E_{\boldsymbol{\alpha}\sim p_0}[\ell_e(\boldsymbol{\alpha};o)]}.
 \label{eq:record-pulse-likelihoods}
\end{align}
We now bound the single-pulse ratios in Eq.~\eqref{eq:record-pulse-likelihoods}.

\medskip
\noindent\textit{Single-pulse likelihood ratio.} Initially assume bounded outcome effects or densities; general outcomes are treated in Sec.~\ref{app:outcomes}. For the fixed pulse and outcome sequence, set $\rho\coloneqq \bigotimes_{j\in B_e}\rho_j(o_{<j})$ and $E_o\coloneqq \bigotimes_{j\in B_e}E_j(o_j\mid o_{<j})$. Collecting the captured displacements into $\bm X_e$ writes Eq.~\eqref{eq:eventlik} as
\begin{align}
 \ell_e(\boldsymbol{\alpha};o)
 =\E_{\boldsymbol{\zeta}}\left[\Tr\left[E_o\hat D(\bm X_e)\rho\hat D^\dagger(\bm X_e)\right]\right],
 \label{eq:pulse-outcome-likelihood}
\end{align}
Here $\bm X_e=M_\tau\boldsymbol{\alpha}+N_\tau\boldsymbol{\zeta}$ is the same captured pulse displacement as in Sec.~\ref{app:window-testing}, now taken over the time bins met by the pulse. The capture matrices have one row for each of these bins. The fixed arrival time and protocol choices are suppressed in the likelihood notation. For the fixed pulse and branch, abbreviate $\ell_e(\cdot;o)$ by $\ell$. Since the two signals differ only in $p_0$ versus $p_{s\eta}$, substituting the amplitude perturbation in Eq.~\eqref{eq:lower-marks} gives the single-pulse probability ratio
\begin{align}
 \frac{\E_{\boldsymbol{\alpha}\sim p_{s\eta}}[\ell(\boldsymbol{\alpha})]}{\E_{\boldsymbol{\alpha}\sim p_0}[\ell(\boldsymbol{\alpha})]}
 =1+4s\ep\frac{\E_{\boldsymbol{\alpha}\sim p_0}[\sin(2\Im\boldsymbol{\alpha}^\dagger\eta)\ell(\boldsymbol{\alpha})]}{\E_{\boldsymbol{\alpha}\sim p_0}[\ell(\boldsymbol{\alpha})]}.
 \label{eq:pulse-likelihood-perturbation}
\end{align}
The ratio therefore depends on the sine feature conditioned on the outcome. With null weight $\mathcal Z[\ell]\coloneqq \E_{\boldsymbol{\alpha}\sim p_0}[\ell(\boldsymbol{\alpha})]>0$, define its Fourier response by
\begin{align}
 G_\eta[\ell]\coloneqq \frac{\E_{\boldsymbol{\alpha}\sim p_0}[e^{\boldsymbol{\alpha}^\dagger\eta-\eta^\dagger\boldsymbol{\alpha}}\ell(\boldsymbol{\alpha})]}{\mathcal Z[\ell]}.
 \label{eq:G}
\end{align}
This is the characteristic function of the pulse amplitudes conditioned on the chosen outcome under the null signal. Its imaginary part is the conditional sine average in Eq.~\eqref{eq:pulse-likelihood-perturbation}, giving the likelihood ratio $1+4s\ep\Im G_\eta[\ell]$. 

\medskip
\noindent\textit{From pulse ratios to the full-record deficit.} Substituting the single-pulse ratio $1+4s\ep\Im G_\eta[\ell_e]$ into Eq.~\eqref{eq:record-pulse-likelihoods} gives
\begin{align}
 L_\eta(o\mid\mathrm{schedule})
 =\frac12\sum_{s=\pm1}\prod_{e=1}^{M}(1+4s\ep\Im G_\eta[\ell_e]).
 \label{eq:scalar-signed-ratio}
\end{align}
Since $|\Im G_\eta[\ell_e]|\le1$ and $\ep\le1/4$, every factor is nonnegative. The arithmetic--geometric mean inequality and $a_{\ep}\coloneqq 1-\sqrt{1-16\ep^2}$ give
\begin{align}
 L_\eta(o\mid\mathrm{schedule})
 \ge\prod_{e=1}^{M}\sqrt{1-16\ep^2(\Im G_\eta[\ell_e])^2}
 \ge1-a_{\ep}\sum_{e=1}^{M}(\Im G_\eta[\ell_e])^2.
 \label{eq:scalar-product-bound}
\end{align}
For the last step, use $\sqrt{1-16\ep^2x}\ge1-a_{\ep}x$ for $0\le x\le1$, followed by $\prod_e(1-x_e)\ge1-\sum_e x_e$. Consequently,
\begin{align}
 (1-L_\eta(o\mid\mathrm{schedule}))_+
 \le a_{\ep}\sum_{e=1}^{M}(\Im G_\eta[\ell_e])^2.
\end{align}
For a fixed pulse and outcome branch, abbreviate $G_\eta[\ell_e]$ by $G_\eta$ and let $\pi_o(\boldsymbol{\alpha})\coloneqq p_0(\boldsymbol{\alpha})\ell_e(\boldsymbol{\alpha};o)/\mathcal Z[\ell_e]$ be the null amplitude distribution conditioned on the outcome. This distribution is independent of $\eta$. Gaussian averaging gives
\begin{align}
 \E_\eta[G_\eta^2]
 =\iint\pi_o(\boldsymbol{\alpha})\pi_o(\boldsymbol{\alpha}')e^{-2\kappa\abs{\boldsymbol{\alpha}+\boldsymbol{\alpha}'}^2}\,d\boldsymbol{\alpha}\,d\boldsymbol{\alpha}'\ge0.
\end{align}
Consequently,
\begin{align}
 \E_\eta[(\Im G_\eta)^2]
 =\frac12\left(\E_\eta[|G_\eta|^2]-\Re\E_\eta[G_\eta^2]\right)
 \le\frac12\E_\eta[|G_\eta|^2].
 \label{eq:sine-overlap}
\end{align}
Taking the null outcome average and using Eqs.~\eqref{eq:record-positive-deficit} and \eqref{eq:sine-overlap} therefore gives
\begin{align}
 \E_\eta\!\left[\TV\!\left(P_0^{\rm tr},\frac{P_{+\eta}^{\rm tr}+P_{-\eta}^{\rm tr}}2\right)\right]
 \le\frac{a_{\ep}}2\E_{\eta,o}\!\left[\sum_{e=1}^{M}|G_\eta[\ell_e]|^2\right].
 \label{eq:scalar-record-deficit}
\end{align}
Here the outcome laws are conditioned on the fixed schedule, and $o$ is averaged under the null, independently of $\eta$. No independence of the pulse responses along the branch is required. We now bound each pulse contribution in the sum. Its fixed arrival time determines whether its signal components lie within individual measurement bins of duration $\tc$ or cross their boundaries. Since the components start $\tc$ apart, they all have the same alignment within a pulse, while different pulses may have different alignments. We use the displacement-overlap bound for noncrossing pulses and the Gaussian-noise bound for crossing pulses. Both give the same bound on $\E_\eta[|G_\eta[\ell_e]|^2]$, uniformly over arrival times and outcome histories.

\medskip
\noindent\textit{Displacement-overlap bound.} We use the following multimode form of the displacement-overlap estimate in Ref.~\cite{oh}, Supplemental Material, Lemma S1. Assume
\begin{align}
 0<\kappa\le\frac{\sigma_{\alpha}^4-1}{4\sigma_{\alpha}^2},\qquad
 b_\kappa\coloneqq 1+\frac{4\kappa\sigma_{\alpha}^2}{1+\sigma_{\alpha}^2}.
 \label{eq:overlap-range}
\end{align}
Here $b_\kappa$ is the inverse one-mode response bound.

For a positive mixture $\ell=\sum_z w_z\ell_z$, the normalized response is $G_\eta[\ell]=\sum_z\pi_zG_\eta[\ell_z]$, where $\pi_z\coloneqq w_z\mathcal Z[\ell_z]/\mathcal Z[\ell]$ and $\sum_z\pi_z=1$, omitting components of zero null weight. If the component weights are independent of $\boldsymbol{\alpha}$ and $\eta$, these null weights are independent of $\eta$, so convexity gives
\begin{align}
 \E_\eta[|G_\eta[\ell]|^2]\le\sum_z\pi_z\E_\eta[|G_\eta[\ell_z]|^2].
 \label{eq:mixture-second-moment}
\end{align}
The same bound holds for continuous mixtures, with sums replaced by integrals. 

\begin{lemma}[Multimode displacement overlap]\label{lem:overlap}
For any $m$-mode density operator $\rho$ and any bounded $E\ge0$, let $g(z)\coloneqq \Tr[E \hat D(z)\rho \hat D^\dagger(z)]$, with $z\in\mathbb C^m$. Under Eq.~\eqref{eq:overlap-range}, if $\mathcal Z[g]>0$, then
\begin{align}
 \E_{\eta\sim\CN(0,2\kappa\Id_m)}[\abs{G_\eta[g]}^2]\le b_\kappa^{-m}. \label{eq:overlap}
\end{align}
Both the input and the effect may be entangled among the $m$ modes.
\end{lemma}
\begin{proof}
For pure input and rank-one effect, this is Ref.~\cite{oh}, Supplemental Material, Lemma S1 and Eq.~(S170), with $2\sigma^2=\sigma_\alpha^{-2}$ and $\sigma_\gamma^2=\kappa$. The general case follows by positive decompositions of $\rho$ and $E$ and convexity in Eq.~\eqref{eq:mixture-second-moment}.
\end{proof}

\medskip
\noindent\textit{Gaussian-noise bound.} Crossing pulses spread the signal amplitudes across extra measured modes. The variance-two background in Eq.~\eqref{eq:noiseidentity} makes their likelihood a positive Gaussian mixture, giving a bound that depends on the signal-amplitude dimension.

\begin{lemma}[Gaussian displacement noise]\label{lem:noise}
Suppose a $d$-mode block is displaced by $H\boldsymbol{\alpha}+z$, with $z\sim\CN(0,2\Id_d)$ independent of $\boldsymbol{\alpha}\in\mathbb C^m$. Let $\lambda_k$ be the eigenvalues of $H^\dagger H$. For any input state, including entanglement with undisplaced auxiliary systems, and any bounded positive joint effect, its likelihood $\ell(\boldsymbol{\alpha})$ satisfies
\begin{align}
 \E_\eta[|G_\eta[\ell]|^2]
 \le\prod_{k=1}^m\frac{1+\sigma_{\alpha}^2\lambda_k}{1+\sigma_{\alpha}^2\lambda_k+4\kappa\sigma_{\alpha}^2}\le b_\kappa^{-m}
 \quad\text{if }\norm H\le1. \label{eq:noisyoverlap}
\end{align}
\end{lemma}
\begin{proof}
Write $\Phi_t$ for additive displacement noise of variance $t$. Splitting $\Phi_2=\Phi_1\circ\Phi_1$ between the state and effect, and using the Gaussian measure-and-prepare identity $\Phi_1(X)=\int\langle z|X|z\rangle\ket z\bra z\,d^{2d}z/\pi^d$~\cite[Theorem~2]{holevo}, gives
\begin{align}
 \ell(\boldsymbol{\alpha})=\iint\frac{d^{2d}z\,d^{2d}w}{\pi^{2d}}
 \Tr[F_w\omega_z]e^{-\abs{H\boldsymbol{\alpha}+z-w}^2}, \label{eq:garbling}
\end{align}
where $\omega_z\coloneqq \langle z|\rho|z\rangle$ and $F_w\coloneqq \langle w|E|w\rangle$ are positive operators on the undisplaced auxiliary systems. Thus the likelihood is a positive mixture of Gaussian components with the same $H$ and weights independent of $\boldsymbol{\alpha}$ and $\eta$.

For each component, the null amplitude law gives posterior covariance $C=(\sigma_\alpha^{-2}\Id_m+H^\dagger H)^{-1}$. Its mean contributes only a phase, so $|G_\eta|^2=e^{-2\eta^\dagger C\eta}$ and
\begin{align}
 \E_\eta[|G_\eta|^2]=\det(\Id_m+4\kappa C)^{-1}
 =\prod_{k=1}^{m}\frac{1+\sigma_\alpha^2\lambda_k}{1+\sigma_\alpha^2\lambda_k+4\kappa\sigma_\alpha^2}.
 \label{eq:gausslik}
\end{align}
Convexity in Eq.~\eqref{eq:mixture-second-moment} passes this bound to $\ell$. Each factor increases with $\lambda_k$, giving $b_\kappa^{-m}$ when $\norm H\le1$.
\end{proof}

\medskip
\noindent\textit{Uniform bound over arrival times.} We now apply the two lemmas to arbitrary pulse arrival times. If every signal component lies within one bin, we use the displacement-overlap bound. Otherwise, the pulse occupies an additional measured mode, and we use the Gaussian-noise bound. To describe these two cases, let $b\in[0,\tc)$ be the time from the first signal onset $s_{e,1}+w$ to the next bin boundary, with $s_{e,1}+w+b\in\tc\mathbb Z$. Since successive components start $\tc$ apart, the same value of $b$ applies to all components.

If $b\notin(0,w)$, each signal component lies within one distinct measurement time bin and its amplitude is captured fully. The noise may spill into neighboring time bins, but conditioning on its entire vector $\boldsymbol{\zeta}$ makes it a deterministic displacement in each time bin. The clean phases in Eq.~\eqref{eq:clean} are a subset of these noncrossing phases. The per-event likelihood bound covers the entire noncrossing case, including phases with nonzero captured background.

If $0<b<w$, each signal component crosses a time-bin boundary. Its left capture fraction is $f\coloneqq b/w$. On the $m+1$ consecutive time bins met by the event,
\begin{align}
 \bm X_e=H_f \boldsymbol{\alpha}+\Delta_g\boldsymbol{\zeta},\qquad
 H_f \boldsymbol{\alpha}\coloneqq f(\boldsymbol{\alpha},0)+(1-f)(0,\boldsymbol{\alpha}). \label{eq:crossing-app}
\end{align}
Indeed the left noise component is wholly in time bin $k$, the right noise component in time bin $k+1$, and the additional noise component in time bin $m+1$. The condition $3w<\tc$ guarantees these containments. Here $(\boldsymbol{\alpha},0)$ and $(0,\boldsymbol{\alpha})$ are the two ways to place the $m$ amplitudes in $m+1$ bins. Since both embeddings preserve the norm, their convex combination satisfies $\abs{H_f\boldsymbol{\alpha}}\le f\abs{\boldsymbol{\alpha}}+(1-f)\abs{\boldsymbol{\alpha}}=\abs{\boldsymbol{\alpha}}$. Thus $\norm{H_f}\le1$. Equations~\eqref{eq:noiseidentity} and \eqref{eq:crossing-app} meet Lemma~\ref{lem:noise} exactly.

\begin{proposition}[Uniform per-event bound]\label{prop:event}
For every arrival phase, the likelihood $\ell_e$ of a fixed product of local measurement effects on an arbitrary event-block input state satisfies
\begin{align}
 \E_{\eta\sim\CN(0,2\kappa\Id_m)}[|G_\eta[\ell_e]|^2]\le b_\kappa^{-m}, \label{eq:eventbound}
\end{align}
provided its null average is finite and positive. The input may be entangled across the bins of the block.
\end{proposition}
\begin{proof}
For noncrossing events, fix $\boldsymbol{\zeta}$, absorb the deterministic background displacements into the input, and apply the square roots of the effects on bins without signal before tracing those bins out. This leaves a positive, possibly unnormalized input on the $m$ signal bins, so Lemma~\ref{lem:overlap} applies after normalization. Averaging over $\boldsymbol{\zeta}$ preserves the bound by Eq.~\eqref{eq:mixture-second-moment}. The crossing case follows from Lemma~\ref{lem:noise} and the capture geometry above.
\end{proof}

This completes the scalar per-pulse estimate, uniformly over arrival times and outcome histories. We now insert it into Eq.~\eqref{eq:scalar-record-deficit} and average the pulse count to obtain the acquisition-time bound.

\subsection{From pulse responses to acquisition time}\label{app:sequential}\label{sec:transcript}\label{sec:time}\label{app:outcomes}
Equation~\eqref{eq:scalar-record-deficit} bounds the conditional TVD by $a_{\ep}/2$ times the null average of the sum of squared pulse responses. For every fixed schedule and outcome branch, Eq.~\eqref{eq:eventbound} bounds each term by $b_\kappa^{-m}$ after averaging over $\eta$. Taking the null outcome average therefore gives
\begin{align}
 \E_\eta\!\left[\TV\!\left(P_0^{\rm tr},\frac{P_{+\eta}^{\rm tr}+P_{-\eta}^{\rm tr}}2\right)\right]
 \le\frac{a_{\ep}}2 M b_\kappa^{-m}.
 \label{eq:sequential-fixed-count}
\end{align}
The estimate is uniform over all adaptive choices fixed by the branch.

\medskip
\noindent\textit{Averaging the pulse count.} Pulses intersecting $[0,\tobs)$ start in $(-L_p,\tobs)$. Campbell's formula gives
\begin{align}
 \E_{\rm schedule}[M]=\rho(\tobs+L_p)=\frac{\tobs+L_p}{3T}.
 \label{eq:event-count}
\end{align}
Averaging over the common schedule and classical randomness, then discarding the revealed schedule, yields
\begin{align}
 \E_\eta\!\left[\TV\!\left(P_0^{\rm tr},\frac{P_{+\eta}^{\rm tr}+P_{-\eta}^{\rm tr}}2\right)\right]
 \le\frac{a_{\ep}}2\rho(\tobs+L_p)b_\kappa^{-m}.
 \label{eq:TVupper}
\end{align}

\medskip
\noindent\textit{Comparison with the testing requirement.} For parameters satisfying Eq.~\eqref{eq:overlap-range}, $q_m>0$, and $\eps<\ep\rho\lc(1-c)$, Eq.~\eqref{eq:TVlower} requires the same average TVD to be at least $q_m/3$. Comparing with Eq.~\eqref{eq:TVupper} gives
\begin{align}
 \frac{\tobs}{T}\ge\frac{2q_m}{a_{\ep}}b_\kappa^m-\frac{L_p}{T}.
 \label{eq:sequential-finite}
\end{align}
To obtain the main-text bound, use Eq.~\eqref{eq:lower-scale-choice}, which gives $b_\kappa=\sigma_\alpha^2$ and $q_m>9/20$, and choose
\begin{align}
 \ep=6n\eps,\qquad 0<\eps\le\frac1{24n}.
 \label{eq:lower-accuracy}
\end{align}
Then $\ep\le1/4$, the feature gap is $\gamma=9\eps/4>2\eps$, and $a_{\ep}\le16\ep^2=576n^2\eps^2$. Substitution into Eq.~\eqref{eq:sequential-finite}, with $m=n/2$ and $L_p/T<1$, proves Theorem~\ref{thm:lower}; the subtraction is absorbed into a universal prefactor for $n\ge8$.

\medskip
\noindent\textit{General outcomes.} For continuous POVMs, project each measured output onto an increasing finite Fock subspace, recording truncation failure as an extra outcome. The resulting POVMs have bounded densities relative to a full-rank reference-state outcome law and retain the original feedback on ordinary outcomes. For finite acquisition, trace-norm convergence under this truncation, contractivity of the subsequent adaptive operations, and dominated convergence~\cite{kallenberg} give TVD convergence of the complete classical records, so the cutoff-independent bounds pass to the limit. The same argument applies to the operator bound in Sec.~\ref{app:coherent-input}.

The sequential bound uses scalar factorization of the full-record likelihood. We next retain temporal input correlations by replacing this factorization with an operator representation.

\subsection{Operator estimate for temporally entangled inputs}\label{app:coherent-input}
We now bound the same full-record TVD for temporally entangled inputs. The pulse amplitudes remain independent between events, but the joint input correlations prevent the scalar likelihood factorization used in Sec.~\ref{app:overlap}. We first state the final TVD bound, then derive its operator likelihood representation and the single-pulse estimate needed to prove it.

\begin{proposition}[Record bound for temporally entangled inputs]\label{prop:transcript}
Under Eq.~\eqref{eq:overlap-range} and $0<\ep\le1/4$, any local idler-free protocol with temporally entangled inputs satisfies
\begin{align}
 \E_\eta\!\left[\TV\!\left(P_0^{\rm tr},\frac{P_{+\eta}^{\rm tr}+P_{-\eta}^{\rm tr}}2\right)\right]
 \le a_{\ep}\rho(\tobs+L_p)\left(\frac{9}{5\sqrt{2b_\kappa}}\right)^m.
 \label{eq:coherent-tv}
\end{align}
Here $\rho=1/(3T)$ is the stationary pulse rate and $a_{\ep}$ is defined in Eq.~\eqref{eq:scalar-product-bound}. The record laws include the protocol's classical randomness.
\end{proposition}
\begin{proof}
We first fix and reveal the schedule, as in Sec.~\ref{app:overlap}, and fix the protocol's classical randomness. We derive a conditional bound for the resulting $M$ pulse blocks, then average over the schedule. Initially assume bounded outcome effects or densities; the approximation in Sec.~\ref{app:sequential} extends the final bound to general outcomes. Fix a complete outcome sequence $o$. This fixes every feedback choice, including all preparation maps and local effects $E_j(o_j\mid o_{<j})$. A later preparation map $\mathcal C_j$ acts only on probes that have not yet encountered the signal. It therefore commutes with an earlier displacement and destructive readout. For an earlier measured bin $i$,
\begin{align}
 \Tr_i\!\left[(E_i\otimes\Id)(\operatorname{id}_i\otimes\mathcal C_j)(\omega)\right]
 =\mathcal C_j\!\left(\Tr_i[(E_i\otimes\Id)\omega]\right).
 \label{eq:coherent-preparation-commutation}
\end{align}
Moving these maps to the beginning in their original order represents the branch by a positive joint input operator $\rho_o$ and the fixed local effects. Both are independent of the signal hypothesis. The operator $\rho_o$ can depend on $o$ and retains the temporal input correlations. This is a representation of the branch probability; it does not change the causal order of the physical protocol. Let $B_e$ be the observed bins met by pulse $e$. The gap $2T-L_p>\tc$ makes these sets disjoint, so the displacement in $j\in B_e$ depends only on $(\boldsymbol\alpha_e,\boldsymbol\zeta_e)$. Let $Q_o$ be the tensor product of the local effects on bins outside all pulse blocks. For fixed amplitudes, the joint-input representation obtained by applying Eq.~\eqref{eq:coherent-preparation-commutation} and the binwise Born rule gives the counterpart of Eq.~\eqref{eq:record-fixed-amplitudes}:
\begin{align}
 P^{\rm tr}\!\left(o\mid\mathrm{schedule},\{\boldsymbol\alpha_e,\boldsymbol\zeta_e\}_{e=1}^M\right)
 =\Tr\!\left[\rho_o\left\{\left(\bigotimes_{e=1}^M\bigotimes_{j\in B_e}
 \hat D^\dagger(X_j)E_j\hat D(X_j)\right)\otimes Q_o\right\}\right].
 \label{eq:coherent-fixed-amplitudes}
\end{align}
Here the displacement arguments and outcome-history dependence of $E_j$ are suppressed only in this equation. Unlike the sequential expression, the trace is taken in one joint input and need not factor into pulse probabilities.

Average each pulse block over its common background, then over its signal amplitudes:
\begin{align}
 L_e(\boldsymbol\alpha;o)\coloneqq\E_{\boldsymbol\zeta}\!\left[\bigotimes_{j\in B_e}
 \hat D^\dagger(X_j(\boldsymbol\alpha,\boldsymbol\zeta))E_j(o_j\mid o_{<j})\hat D(X_j(\boldsymbol\alpha,\boldsymbol\zeta))\right],\qquad F_{e,\theta}\coloneqq\int p_\theta(\boldsymbol\alpha)L_e(\boldsymbol\alpha;o)d\boldsymbol\alpha.
 \label{eq:event-effects}
\end{align}
For a pulse-block input $\rho$, $\Tr[\rho L_e(\boldsymbol\alpha;o)]$ is the scalar single-pulse likelihood. We keep $L_e$ as an operator so that input correlations with other blocks remain available. Identity effects complete unobserved bins, as in Sec.~\ref{app:overlap}. The classical amplitudes of different pulses are independent. Their averages thus factor inside the trace, giving
\begin{align}
 P_\theta^{\rm tr}(o\mid\mathrm{schedule})
 =\Tr\!\left[\rho_o\left\{\left(\bigotimes_{e=1}^M F_{e,\theta}\right)\otimes Q_o\right\}\right].
 \label{eq:coherent-record-probability}
\end{align}
This is the analogue of Eq.~\eqref{eq:record-pulse-factorization}: the pulse effects factor, while the joint input expectation retains their quantum correlations.

Substituting the signed amplitude law $p_{s\eta}(\boldsymbol\alpha)=p_0(\boldsymbol\alpha)[1+4s\ep\sin(2\Im\boldsymbol\alpha^\dagger\eta)]$ gives
\begin{align}
 F_{e,s\eta}-F_{e,0}
 =4s\ep\int p_0(\boldsymbol\alpha)\sin(2\Im\boldsymbol\alpha^\dagger\eta)L_e(\boldsymbol\alpha;o)d\boldsymbol\alpha.
\end{align}
To form the record likelihood ratio, we express each signed pulse effect relative to its null effect. The scalar calculation gives the normalized factor $1+4s\ep\Im G_\eta[\ell_e]$. For operators, multiplication by $F_{e,0}^{-1/2}$ on both sides normalizes the null effect to the identity and preserves Hermiticity. Accordingly, define $S_{e,\eta}$ on the support of $F_{e,0}$ by
\begin{align}
 S_{e,\eta}\coloneqq F_{e,0}^{-1/2}\left[\int p_0(\boldsymbol\alpha)\sin(2\Im\boldsymbol\alpha^\dagger\eta)L_e(\boldsymbol\alpha;o)d\boldsymbol\alpha\right]F_{e,0}^{-1/2}.
 \label{eq:event-sine-response}
\end{align}
The dependence of $S_{e,\eta}$ on the fixed branch $o$ and schedule is left implicit. Since $L_e\ge0$ and the sine lies in $[-1,1]$, the integral in brackets lies between $-F_{e,0}$ and $F_{e,0}$. Hence
\begin{align}
 -\Id\preceq S_{e,\eta}\preceq\Id,\qquad
 F_{e,s\eta}=F_{e,0}^{1/2}(\Id+4s\ep S_{e,\eta})F_{e,0}^{1/2}.
 \label{eq:event-normalization}
\end{align}
Throughout, inverse square roots are taken on the support of the null effect, and each normalized product denotes the bounded extension of the associated form. Zero null mass gives zero alternative mass because $p_{s\eta}\le(1+4\ep)p_0$.

In the full-record trace, the two sets of null factors can now be absorbed into the input. Define
\begin{align}
 K_o\coloneqq\left(\bigotimes_{e=1}^M F_{e,0}^{1/2}\right)\otimes Q_o^{1/2},\qquad
 \sigma_o\coloneqq\frac{K_o\rho_oK_o}{\Tr(K_o\rho_oK_o)}.
\end{align}
The null probability is $P_0^{\rm tr}(o\mid\mathrm{schedule})=\Tr(K_o\rho_oK_o)$, so for a branch of positive null probability, $\sigma_o$ is a normalized joint state independent of $\eta$. To see how the null effects cancel in the ratio, substitute Eq.~\eqref{eq:event-normalization} into the tensor product in Eq.~\eqref{eq:coherent-record-probability}. Since each pulse occupies a disjoint block,
\begin{align}
 \left(\bigotimes_{e=1}^M F_{e,s\eta}\right)\otimes Q_o
 &=\left[\left(\bigotimes_{e=1}^M F_{e,0}^{1/2}\right)\otimes Q_o^{1/2}\right]
 \left[\prod_{e=1}^M(\Id+4s\ep S_{e,\eta})\right]
 \left[\left(\bigotimes_{e=1}^M F_{e,0}^{1/2}\right)\otimes Q_o^{1/2}\right]\\
 &=K_o\left[\prod_{e=1}^M(\Id+4s\ep S_{e,\eta})\right]K_o.
 \label{eq:coherent-null-factorization}
\end{align}
Here each $S_{e,\eta}$ is extended by the identity on all other pulse blocks and on the bins outside events. Inserting this factorization into the record probability and cycling the trace gives
\begin{align}
 P_{s\eta}^{\rm tr}(o\mid\mathrm{schedule})=\Tr\!\left[\rho_oK_o\left\{\prod_e(\Id+4s\ep S_{e,\eta})\right\}K_o\right]=\Tr\!\left[K_o\rho_oK_o\prod_e(\Id+4s\ep S_{e,\eta})\right].
 \label{eq:coherent-trace-normalization}
\end{align}
For the null hypothesis the middle product is the identity, giving $P_0^{\rm tr}(o\mid\mathrm{schedule})=\Tr(K_o\rho_oK_o)$. Dividing by this probability and using the definition of $\sigma_o$ therefore yields
\begin{align}
 \frac{P_{s\eta}^{\rm tr}(o\mid\mathrm{schedule})}{P_0^{\rm tr}(o\mid\mathrm{schedule})}
 =\Tr\!\left[\sigma_o\prod_{e=1}^M(\Id+4s\ep S_{e,\eta})\right].
 \label{eq:coherent-record}
\end{align}
All input correlations remain in $\sigma_o$. The sign is chosen once for the full acquisition, so its average is outside the product:
\begin{align}
 L_\eta(o\mid\mathrm{schedule})
 =\frac12\sum_{s=\pm1}\Tr\!\left[\sigma_o\prod_{e=1}^M(\Id+4s\ep S_{e,\eta})\right].
 \label{eq:signed-record-ratio}
\end{align}

Operators on distinct pulse blocks commute. Their joint spectral values belong to $[-1,1]$, so the scalar signed-product inequality in Eq.~\eqref{eq:scalar-product-bound} gives the operator inequality
\begin{align}
 \frac12\sum_{s=\pm1}\prod_{e=1}^M(\Id+4s\ep S_{e,\eta})
 \succeq\Id-a_{\ep}\sum_{e=1}^M S_{e,\eta}^2.
\end{align}
Taking its expectation in $\sigma_o$, and then the positive deficit, yields
\begin{align}
 (1-L_\eta(o\mid\mathrm{schedule}))_+
 \le a_{\ep}\sum_{e=1}^M\Tr(\sigma_oS_{e,\eta}^2).
 \label{eq:branch-likelihood-deficit}
\end{align}
The null outcome average of this deficit is the conditional TVD by Eq.~\eqref{eq:record-positive-deficit}. Therefore
\begin{align}
 \E_\eta\!\left[\TV\!\left(P_0^{\rm tr},\frac{P_{+\eta}^{\rm tr}+P_{-\eta}^{\rm tr}}2\right)\right]
 \le a_{\ep}\E_{\eta,o}\!\left[\sum_{e=1}^M\Tr(\sigma_oS_{e,\eta}^2)\right].
 \label{eq:record-master}
\end{align}
The laws in this equation are conditional on the fixed schedule and randomness, and $o$ is averaged under the null. If $\sigma_o$ factors over pulse blocks, the trace in Eq.~\eqref{eq:coherent-record} factors into scalar responses $\Tr(\sigma_oS_{e,\eta})$; applying the scalar inequality after this factorization recovers the sequential calculation. For a general entangled input we must instead bound the second moments in Eq.~\eqref{eq:record-master}.

\medskip
\noindent\textit{The remaining second-moment bound.} Equation~\eqref{eq:record-master} requires a bound on $\Tr(\sigma_oS_{e,\eta}^2)$ for each pulse. The scalar estimate from Sec.~\ref{app:overlap} controls the square of an expectation, whereas the present expression contains the expectation of a square:
\begin{align}
 [\Tr(\sigma_oS_{e,\eta})]^2\le\Tr(\sigma_oS_{e,\eta}^2).
\end{align}
An upper bound on the left-hand side does not bound the right-hand side. The following lemma gives the required operator bound using Proposition~\ref{prop:event} for every pure input of the same pulse block. We first apply it to finish the proof of Proposition~\ref{prop:transcript}, then give its proof.

\begin{lemma}[Single-pulse operator bound]\label{lem:coherent-lifting}
For each pulse $e$ and fixed outcome sequence $o$, let $L_e(\boldsymbol\alpha;o)$ and $F_{e,0}$ be the pulse effects in Eq.~\eqref{eq:event-effects}, and let $S_{e,\eta}$ be their normalized sine response in Eq.~\eqref{eq:event-sine-response}. Under Eq.~\eqref{eq:overlap-range},
\begin{align}
 \E_{\eta\sim\CN(0,2\kappa\Id_m)}[S_{e,\eta}^2]\preceq\left(\frac{9}{5\sqrt{2b_\kappa}}\right)^m\Id
 \label{eq:coherent-event}
\end{align}
on the support of $F_{e,0}$, uniformly over arrival times and outcome histories.
\end{lemma}

Applying Eq.~\eqref{eq:coherent-event} to each pulse block and using the normalization and $\eta$-independence of $\sigma_o$ gives
\begin{align}
 \E_\eta[\Tr(\sigma_oS_{e,\eta}^2)]
 =\Tr\!\left[\sigma_o\E_\eta[S_{e,\eta}^2]\right]
 \le\left(\frac{9}{5\sqrt{2b_\kappa}}\right)^m.
\end{align}
Inserting this into Eq.~\eqref{eq:record-master} and taking the null outcome average gives the conditional bound
\begin{align}
 \E_\eta\!\left[\TV\!\left(P_0^{\rm tr},\frac{P_{+\eta}^{\rm tr}+P_{-\eta}^{\rm tr}}2\right)\right]
 \le a_{\ep}M\left(\frac{9}{5\sqrt{2b_\kappa}}\right)^m.
\end{align}
Finally, average over classical randomness and the common stationary schedule. Equation~\eqref{eq:event-count} gives $\E[M]=\rho(\tobs+L_p)$. Discarding the revealed schedule can only decrease TVD, proving Eq.~\eqref{eq:coherent-tv} for the original record.
\end{proof}

\begin{proof}[Proof of Lemma~\ref{lem:coherent-lifting}]
Fix the pulse $e$ and outcome sequence $o$. The proof uses the null normalization in Eq.~\eqref{eq:event-sine-response} with the sine replaced by Fourier and Gaussian weights. For a bounded measurable weight $f$, denote the resulting normalized average on the support of $F_{e,0}$ by $\Phi_e(f)$, defined by
\begin{align}
 \Phi_e(f)\coloneqq F_{e,0}^{-1/2}\left[\int p_0(\boldsymbol\alpha)f(\boldsymbol\alpha)L_e(\boldsymbol\alpha;o)d\boldsymbol\alpha\right]F_{e,0}^{-1/2},\qquad \Phi_e(1)=\Id.
 \label{eq:lift-normalized-average}
\end{align}
The sine weight recovers $S_{e,\eta}$, while the complex exponential gives the Fourier response needed to apply Proposition~\ref{prop:event}:
\begin{align}
 R_{e,\eta}\coloneqq\Phi_e(e^{\boldsymbol\alpha^\dagger\eta-\eta^\dagger\boldsymbol\alpha}),\qquad \Im R_{e,\eta}=\Phi_e(\sin(2\Im\boldsymbol\alpha^\dagger\eta))=S_{e,\eta}.
 \label{eq:lift-fourier-sine}
\end{align}
Here $\Im R=(R-R^\dagger)/(2i)$. Cauchy--Schwarz for the defining integral gives a bounded form with $\norm{\Phi_e(f)}\le\norm f_\infty$ on the support of $F_{e,0}$. In particular, $\norm{R_{e,\eta}}\le1$, and the phase symmetry gives $R_{e,-\eta}=R_{e,\eta}^\dagger$. All operator integrals are understood weakly. For every pure pulse-block input $\ket\psi$ with positive null weight, $\bra\psi L_e(\boldsymbol\alpha;o)\ket\psi$ is the single-pulse likelihood of Proposition~\ref{prop:event}. Its normalized Fourier response is the ratio below, so that proposition gives
\begin{align}
 \E_\eta\!\left[\left|\frac{\bra\psi F_{e,0}^{1/2}R_{e,\eta}F_{e,0}^{1/2}\ket\psi}{\bra\psi F_{e,0}\ket\psi}\right|^2\right]\le b_\kappa^{-m}.
 \label{eq:lift-hypothesis}
\end{align}
Every unit vector $v$ in the support of $F_{e,0}$ is a norm limit of normalized vectors in the range of $F_{e,0}^{1/2}$. Since $\norm{R_{e,\eta}}\le1$, dominated convergence extends Eq.~\eqref{eq:lift-hypothesis} to all such $v$, giving
\begin{align}
 \E_\eta[|\bra vR_{e,\eta}\ket v|^2]\le b_\kappa^{-m}\qquad(\langle v|v\rangle=1).
\end{align}

For $h_x(\boldsymbol{\alpha})\coloneqq e^{-4\kappa\abs{\boldsymbol{\alpha}-x}^2/3}$, the Gaussian Fourier representation and Cauchy--Schwarz give, for every unit $v$,
\begin{align}
 \bra v\Phi_e(h_x)\ket v
 &=\E_\eta\!\left[\left(\frac32\right)^m e^{-\abs\eta^2/(4\kappa)}e^{-x^\dagger\eta+\eta^\dagger x}\bra vR_{e,\eta}\ket v\right]
 \\&\le\left(\frac32\right)^m\!\left(\E_\eta e^{-\abs\eta^2/(2\kappa)}\right)^{1/2}\!\left(\E_\eta|\bra vR_{e,\eta}\ket v|^2\right)^{1/2}\le\left(\frac{3}{\sqrt{8b_\kappa}}\right)^m\eqqcolon q.
\end{align}
Thus $0\preceq\Phi_e(h_x)\preceq q\Id$. Applying Cauchy--Schwarz inside the defining amplitude integral, with the factors $h_x^{1/2}$ and $h_x^{5/2}$, gives for unit vectors $u,v$
\begin{align}
 |\bra u\Phi_e(h_x^3)\ket v|^2\le\bra u\Phi_e(h_x)\ket u\,\bra v\Phi_e(h_x^5)\ket v\le q\bra v\Phi_e(h_x^5)\ket v.
 \label{eq:lift-weighted-cauchy-schwarz}
\end{align}
This follows first on the dense range of $F_{e,0}^{1/2}$ and then by continuity. Taking the supremum over $u$ yields $\Phi_e(h_x^3)^2\preceq q\Phi_e(h_x^5)$.

The Gaussian Fourier transform of $h_x^3$ gives
\begin{align}
 \Phi_e(h_x^3)=\frac{1}{(4\pi\kappa)^m}\int_{\mathbb C^m}e^{-\abs\eta^2/(4\kappa)}e^{-x^\dagger\eta+\eta^\dagger x}R_{e,\eta}\,d^{2m}\eta.
\end{align}
Plancherel's identity, applied after acting on any vector $v$, gives the first equality below; the Cauchy--Schwarz bound gives the inequality:
\begin{align}
 \E_\eta[R_{e,\eta}^\dagger R_{e,\eta}]=\left(\frac{8\kappa}{\pi}\right)^m\int_{\mathbb C^m}\Phi_e(h_x^3)^2\,d^{2m}x\preceq q\left(\frac{8\kappa}{\pi}\right)^m\int_{\mathbb C^m}\Phi_e(h_x^5)\,d^{2m}x=q\left(\frac65\right)^m\Id=\left(\frac{9}{5\sqrt{2b_\kappa}}\right)^m\Id.
 \label{eq:gaussian-lifting}
\end{align}
Here $\int_{\mathbb C^m}h_x^5\,d^{2m}x=(3\pi/(20\kappa))^m$ and $\Phi_e(1)=\Id$. Symmetry under $\eta\mapsto-\eta$ gives the same bound for $\E_\eta[R_{e,\eta}R_{e,\eta}^\dagger]$. Since $S_{e,\eta}=\Im R_{e,\eta}$, the inequality $S_{e,\eta}^2\preceq\tfrac12(R_{e,\eta}^\dagger R_{e,\eta}+R_{e,\eta}R_{e,\eta}^\dagger)$ proves Eq.~\eqref{eq:coherent-event}.
\end{proof}

\medskip
\noindent\textit{Acquisition-time lower bound.} 
For parameters satisfying Eq.~\eqref{eq:overlap-range}, $q_m>0$, and $\eps<\ep\rho\lc(1-c)$, comparing with the testing requirement~\eqref{eq:TVlower} yields
\begin{align}
 \frac{\tobs}{T}\ge\frac{q_m}{a_{\ep}}\left(\frac{50b_\kappa}{81}\right)^{m/2}-\frac{L_p}{T}.
 \label{eq:coherent-finite}
\end{align}
With the main-text choices in Eqs.~\eqref{eq:lower-scale-choice} and \eqref{eq:lower-accuracy}, $q_m>9/20$, $a_{\ep}\le576n^2\eps^2$, and $50b_\kappa/81\ge200/81>1$. The subtraction $L_p/T<1$ is absorbed into a universal prefactor for $n\ge8$, proving Theorem~\ref{thm:coherent-input}.

Section~\ref{app:source-checks} verifies the remaining power and dependence conditions for the stated stationary signal class.

\medskip
\noindent\textit{Remark (smaller query scales).} For any fixed $\kappa>0$, the variance choice in Eq.~\eqref{eq:lower-scale-choice} gives $b_\kappa=\sigma_\alpha^2>1$. The subtracted term in Eq.~\eqref{eq:qgap-app} then vanishes as $m\to\infty$, so the sequential bound has the same scaling as Theorem~\ref{thm:lower} for even $n$ above a $\kappa$-dependent threshold.

\subsection{Signal power and temporal dependence}\label{app:source-checks}\label{app:source-mixing}
We verify finite mean power and a dependence time of order $T$, uniformly over the signal-law index $\theta$, as required in Theorem~\ref{thm:lower}.

The components in Eq.~\eqref{eq:event} have disjoint supports, with $\E_\theta[\abs{\boldsymbol\alpha}^2]=m\sigma_\alpha^2$ and $\E[|\zeta_k|^2]=2k$. Each of the first $m$ background components contributes twice, and the last contributes once. Since distinct pulses do not overlap, Campbell's formula gives
\begin{align}
 \E_\theta[|\xi_\theta(t)|^2]=\rho\,\E_\theta\!\left[\int_{\mathbb R}|\xi_e(u)|^2du\right]=\frac8{3n\tc}\left[m\sigma_\alpha^2+2(m+1)^2\right]=O\!\left(\frac{n+\sigma_\alpha^2}{\tc}\right).
 \label{eq:source-power}
\end{align}
Thus the hard family has finite power with at most polynomial growth at fixed $\kappa$. We next bound the dependence penalty entering the Bell estimator.

\begin{lemma}[A uniform physical-time dependence bound]\label{lem:mix}
Let $\xi_{\le t}$ and $\xi_{\ge t'}$ denote the field histories before $t$ and after $t'$, specified by their integrals over finite intervals. For random elements $U,V$, write $\beta_{\rm mix}(U,V)\coloneqq \TV(\operatorname{Law}(U,V),\operatorname{Law}(U)\otimes\operatorname{Law}(V))$. There is a coupling time $C\ge0$, with $\E[C]=T(e^2-1)$, such that for every $u\ge0$,
\begin{align}
 \beta_{\rm mix}(\xi_{\le t},\xi_{\ge t+u})\le\Prob(C>u).
 \label{eq:mix}
\end{align}
Consequently the dependence time of the duration-$T$ windows used in Eq.~\eqref{eq:estimator} satisfies
\begin{align}
 \frac{\tau_{\rm dep}}T\le4e^2<30.
 \label{eq:source-dependence-constant}
\end{align}
\end{lemma}
\begin{proof}
By the standard coupling inequality~\cite{kallenberg}, Eq.~\eqref{eq:mix} follows once a stationary copy independent of the past agrees with the original field after time $t+C$. Construct both schedules from common future Poisson proposals of rate $1/T$, accepting each proposal if at least $2T$ has elapsed since the last accepted start, and attach common fresh amplitudes. Let $C$ be the first time after $t$ when a proposal-free interval of length $2T$ has elapsed entirely after $t$. Both schedules then accept the same subsequent proposals. Since $L_p<2T$, all earlier pulses have ended, so their future fields agree.

The number of proposal intervals needed to obtain one of length at least $2T$ is geometric, giving $\E[C]<\infty$. Conditioning on the first proposal time $E\sim\operatorname{Exp}(1/T)$ gives
\begin{align}
 \E[C]=\E[\min(E,2T)]+\Prob(E<2T)\E[C]=T(1-e^{-2})+(1-e^{-2})\E[C],
\end{align}
so $\E[C]=T(e^2-1)$. Windows $g$ steps apart have gap $(g-1)T$. Equations~\eqref{eq:mix} and \eqref{eq:mixsum}, TVD contractivity, and the standard tail-integral formula~\cite{kallenberg} give
\begin{align}
 \frac{\tau_{\rm dep}}T\le4\sum_{g=1}^\infty\beta_{\rm mix}(g)\le4\sum_{k=0}^\infty\Prob(C>kT)\le4\left(1+\frac{\E[C]}T\right)=4e^2.
\end{align}
The construction is uniform over the amplitude law and includes partial captures at window boundaries.
\end{proof}
Equations~\eqref{eq:source-power} and \eqref{eq:source-dependence-constant} place the hard family within the signal class of Theorem~\ref{thm:lower}. A protocol satisfying the learning guarantee for every signal in this class also satisfies it on the hard family. Its acquisition-time bounds therefore prove Theorems~\ref{thm:lower} and \ref{thm:coherent-input}. The same dependence bound also makes Theorem~\ref{thm:upper} applicable to every member of the family.

\section{Smooth pulses and finite-response sensors}\label{app:smooth-response}
We now verify that the output hard family can be produced by a stationary finite-power physical signal through a sensor with finite response time. We first replace the rectangular profiles in Appendix~\ref{app:lower-proof} by smooth pulses preserving the feature-gap and pulse-response bounds. We then construct an input $S_\theta$ satisfying $h*S_\theta=\xi_\theta$ for a sufficient class of fixed, stable responses $h$, retaining the temporal-mode interface of Sec.~\ref{sec:protocols}.

\subsection{Smooth pulse profiles}\label{app:smooth-profiles}
For the desired output profiles, choose a fixed nonnegative smooth function $\varphi$ on $\mathbb R$, supported in $[0,1]$, with $\int_0^1\varphi(u)du=1$. All endpoint derivatives vanish. Define the component profile
\begin{align}
 u_w(t)\coloneqq \frac{\sqrt\tc}{w}\varphi(t/w),\qquad \tc^{-1/2}\int_{\mathbb R}u_w(t)dt=1.
 \label{eq:smooth-profile}
\end{align}
Keep the arrival process, amplitude laws, group starts $s_{e,k}=\tau_e+(k-1)\tc$, and width $w=\tc/8$ of Appendix~\ref{app:source}. Replacing each rectangular component by this profile gives
\begin{align}
 \xi_e(t)\coloneqq \sum_{k=1}^m\left\{\zeta_k[u_w(t-s_{e,k})-u_w(t-s_{e,k}-2w)]+\alpha_k u_w(t-s_{e,k}-w)\right\}+\zeta_F u_w(t-\tau_e-m\tc).
 \label{eq:smooth-event}
\end{align}
The support remains within $[\tau_e,\tau_e+L_p]$, with $L_p=m\tc+w<T$. A component captured fully in one bin contributes its amplitude exactly. The resulting field $\xi_\theta(t)=\sum_e\xi_e(t)$ is smooth and strictly stationary, since the locally finite pulse construction commutes with time translations.

\medskip
\noindent\textit{Preserving the window feature.}
The capture matrices in Sec.~\ref{app:window-testing} now use profile integrals in place of overlap lengths. For example,
\begin{align}
 (M_\tau)_{jk}\coloneqq \int_{j\tc}^{(j+1)\tc}\frac1w\varphi\left(\frac{t-\tau-(k-1)\tc-w}{w}\right)dt.
 \label{eq:smooth-capture}
\end{align}
The background entries are the corresponding differences of captured fractions, with one final-background column. Thus the fixed-arrival average and Campbell formula, Eqs.~\eqref{eq:fixed-arrival-fourier} and \eqref{eq:campbell}, still hold. On the same clean interval of length $\lc=\tc-3w$, each group is fully captured in one bin, its paired backgrounds cancel, and the final background lies outside the selected block. Hence Eq.~\eqref{eq:clean-capture} holds exactly.

All signal components have the same profile and are spaced by $\tc$. Consequently, the selected capture matrix still has the form $M_{\tau,\mathcal J}=fS_d+(1-f)S_{d+1}$ for some integer $d$ and $0\le f\le1$, including truncation at the window edges. Thus Eq.~\eqref{eq:detgap} and Proposition~\ref{prop:gap} give the same feature-gap and testing bounds, with unchanged $q_m$, $\gamma$, and $\lc$.

\medskip
\noindent\textit{Preserving the pulse-response bounds.}
If a signal component crosses a bin boundary at a time $b\in(0,w)$ after its onset, its left capture fraction becomes
\begin{align}
 f\coloneqq \int_0^{b/w}\varphi(u)du\in[0,1],\qquad
 \bm X_e=H_f\boldsymbol{\alpha}+\Delta_g\boldsymbol{\zeta}.
 \label{eq:smooth-crossing}
\end{align}
The preceding and following background components remain wholly in adjacent bins, so their captured covariance is still $2\Id_{m+1}$. Since $H_f\boldsymbol{\alpha}=f(\boldsymbol{\alpha},0)+(1-f)(0,\boldsymbol{\alpha})$, we again have $\norm{H_f}\le1$, and Lemma~\ref{lem:noise} applies. For noncrossing components, the signal amplitudes are captured fully; conditioning on the background gives the same displacement-overlap argument. Thus Proposition~\ref{prop:event} and the operator estimate in Eq.~\eqref{eq:coherent-event} hold with their original constants. The pulse-count and likelihood-ratio arguments are unchanged.

The power calculation in Eq.~\eqref{eq:source-power} gains only the profile factor $c_0\coloneqq\int_0^1\varphi(u)^2du$:
\begin{align}
 \E_\theta[|\xi_\theta(t)|^2]=\rho\frac{\tc}{w}c_0\left[m\sigma_\alpha^2+2(m+1)^2\right].
 \label{eq:smooth-output-power}
\end{align}
The arrival law and pulse support are unchanged, so Lemma~\ref{lem:mix} still gives $\tau_{\rm dep}\le4e^2T$. The same acquisition-time lower bounds and Bell upper bound therefore apply.

\subsection{A sufficient condition on the sensor response}\label{app:response-realization}

We now specify a class of sensor responses for which the desired output can be realized by a stationary finite-power input. The complex input $S_\theta$ represents two real components coupled to the two output quadratures, as allowed by the signal model in Sec.~\ref{sec:setup}. For example, a damped cavity amplitude on resonance in a rotating frame has the unit-static-gain response $h(t)=\tau_s^{-1}e^{-t/\tau_s}\mathbf1_{t\ge0}$, with fixed response time $\tau_s>0$. Its Laplace transfer function and the input chosen to produce the prescribed smooth output are
\begin{align}
 H(z)=\frac1{1+\tau_s z},\qquad S_\theta\coloneqq \xi_\theta+\tau_s\dot\xi_\theta.
 \label{eq:single-pole-input}
\end{align}
The derivative term compensates the sensor relaxation. More generally, a sufficient condition is that the calibrated response have Laplace transfer function
\begin{align}
 H(z)=\frac{g}{P(z)},\qquad P(z)\coloneqq \sum_{\ell=0}^q p_\ell z^\ell,\qquad g\ne0,
 \label{eq:response-sufficient}
\end{align}
where $P$ has degree $q\ge1$, every root of $P$ has strictly negative real part, and the response parameters are fixed independently of $n$. This includes a single relaxation response, cascades of such responses, and damped-oscillator responses with constant numerator. The causal impulse response is stable. Since $P(z)H(z)=g$, the prescribed output is produced by the input
\begin{align}
 S_\theta(t)\coloneqq \frac1g\sum_{\ell=0}^q p_\ell\frac{d^\ell\xi_\theta(t)}{dt^\ell}.
 \label{eq:physical-smooth-input}
\end{align}
The smooth component profiles have compact support and vanishing endpoint derivatives, so integration by parts gives
\begin{align}
 (h*S_\theta)(t)=\xi_\theta(t).
 \label{eq:physical-smooth-response}
\end{align}
Differentiation commutes with time translations, so $S_\theta$ is strictly stationary. Its pulse components have the same support as those of $\xi_\theta$, and the stationary convolution is well defined by the finite-power calculation below and absolute integrability of $h$.

To bound the input power, write $c_\ell\coloneqq\int_0^1|\varphi^{(\ell)}(u)|^2du<\infty$. The rescaled profiles satisfy $\|d^\ell u_w/dt^\ell\|_{L^2}^2=\tc c_\ell/w^{2\ell+1}$. Cauchy--Schwarz bounds the energy of their weighted sum, and Campbell's formula converts it to the input power:
\begin{align}
 \E_\theta[|S_\theta(t)|^2]\le\frac{\rho(q+1)\tc}{|g|^2}\left[m\sigma_\alpha^2+2(m+1)^2\right]\sum_{\ell=0}^q\frac{|p_\ell|^2c_\ell}{w^{2\ell+1}}.
 \label{eq:smooth-input-power-bound}
\end{align}
Thus the input has finite mean power uniformly over the hypotheses. For fixed $\kappa$, profile, and sensor response, this power is $O(n)$ at fixed $\tc$, and $O(n^{2q+2})$ at fixed $T=n\tc$. The extension therefore requires only polynomial signal power. The same renewal coupling also bounds the dependence of the input, since differentiation leaves the pulse supports unchanged.

The realized output is exactly the smooth hard family. Theorems~\ref{thm:lower} and \ref{thm:coherent-input} and the Bell upper bound therefore retain their acquisition-time rates for its output window statistics.

\end{document}